\documentclass{article}

\usepackage[affil-it]{authblk}
\usepackage{amsfonts}
\usepackage{amsmath,amsthm,amssymb}
\usepackage{enumerate}
\usepackage[english]{babel}
\usepackage{graphicx}
\usepackage[margin=2.8cm]{geometry}

\usepackage{hyperref}
\hypersetup{colorlinks=true,citecolor=blue,linkcolor=blue,filecolor=blue,urlcolor=blue}

\theoremstyle{plain}
\newtheorem{theorem}{Theorem}[section]
\newtheorem{lemma}[theorem]{Lemma}

\newtheorem{corollary}[theorem]{Corollary}

\theoremstyle{definition}
\newtheorem{definition}[theorem]{Definition}
\newtheorem{remark}[theorem]{Remark}

\newcommand*{\cI}{\mathcal{I}}

\newcommand*{\cT}{\mathcal{T}}

\newcommand*{\cW}{\mathcal{W}}

\newcommand*{\eps}{\varepsilon}

\newcommand*{\id}{\mathrm{id}}
\newcommand*{\tr}{\mathrm{tr}}
\newcommand*{\ket}[1]{| #1 \rangle}
\newcommand*{\bra}[1]{\langle #1 |}
\newcommand*{\spr}[2]{\langle #1 | #2 \rangle}
\newcommand*{\proj}[1]{|#1\rangle\!\langle #1|}
\newcommand*{\Sym}{\mathrm{Sym}}

\newcommand{\ceil}[1]	{\left\lceil #1 \right\rceil}

\newcommand*{\comment}[1] {}

\begin{document}

\title{Quantum conditional mutual information \\ and approximate Markov chains}

\author[1,2]{Omar Fawzi}
\author[1]{Renato Renner}
\affil[1]{ETH Zurich, Switzerland}
\affil[2]{LIP\footnote{UMR 5668 ENS Lyon - CNRS - UCBL - INRIA, Université de Lyon.}, ENS de Lyon, France}

\date{}

\maketitle

\begin{abstract}
 A state on a tripartite quantum system $A \otimes B \otimes C$ forms a Markov chain if it can be reconstructed from its marginal on $A \otimes B$ by a quantum operation from $B$ to $B \otimes C$.  We show that the quantum conditional mutual information $I(A: C | B)$ of an arbitrary state is an upper bound on its distance to the closest reconstructed state. It thus quantifies how well the Markov chain property is approximated. 
  \end{abstract}

\section{Introduction} \label{sec_intro}

The \emph{conditional mutual information} $I(A : C | B)_{\rho} = H(\rho_{A B}) + H(\rho_{B C}) - H(\rho_B) - H(\rho_{A B C})$ of a state $\rho_{ABC}$ on a tripartite system $A \otimes B \otimes C$ is meant to quantify the correlations between $A$ and $C$ from the point of view of $B$. Here $H(\rho) = -\tr(\rho \log_2 \rho)$ is the von Neumann entropy. Apart from its central role in traditional information theory, the conditional mutual information has recently found applications in new areas of computer science and physics. Examples include communication and information complexity (see~\cite{Bra12} and references therein), de Finetti type theorems~\cite{BH13, BH13b} and also the study of quantum many-body systems~\cite{Kim13thesis}. The importance of the conditional mutual information for such applications is due to its various useful properties. In particular, it has an additivity property called the \emph{chain rule}: ${I(A_1 \dots A_n : C | B)} = {I(A_1 : C |B)} + {I(A_2 : C | B A_1)} + \dots + {I(A_n : C | B A_1 \dots A_{n-1})}$.

When the $B$ system is classical, the conditional mutual information $I(A : C | B)$ has a simple interpretation: it is the average over the values $b$ taken by $B$ of the (unconditional) mutual information evaluated for the conditional state on the system $A \otimes C$. This is crucial for applications because the (unconditional) mutual information can be related to operational quantities such as the distance to product states using Pinsker's inequality for instance. However, when $B$ is quantum, the conditional mutual information is significantly more complicated and much less is known about it. In fact, even the fact that $I(A:C|B) \geq 0$, also known as \emph{strong subadditivity} of the von Neumann entropy, is a highly non-trivial theorem~\cite{LieRus73}. The structure of states that satisfy $I(A:C|B)_{\rho} = 0$ was also studied~\cite{Pet88,HJPW04}. It has been found that a zero conditional mutual information characterises states $\rho_{A B C}$ whose $C$ system can be reconstructed just by acting on $B$, i.e., there exists a quantum operation~$\cT_{B \to BC}$ from the $B$ to the $B \otimes C$ system such that 
\begin{align}
\label{eq_markovchain}
\rho_{ABC} = \cT_{B \to BC}(\rho_{AB}) \ .
\end{align}
States $\rho_{A B C}$ that satisfy this condition are called \emph{(quantum) Markov chains}. When $B$ is classical the condition~\eqref{eq_markovchain} simply means that, for all values $b$ taken by $B$, the conditional state on $A \otimes C$ is a product state. We say that \emph{$A$ and $C$ are independent given $B$}. 

A natural question that is very relevant for applications is to characterise states for which the conditional mutual information is approximately zero, i.e., for which it is guaranteed that $I(A:C|B) \leq \epsilon$ for some $\epsilon > 0$. In applications  involving $n$ systems $A_1, \ldots, A_n$, such a guarantee is often obtained from an upper bound on the total conditional mutual information ${I(A_1 \dots A_n : C | B)} \leq c$ (which can even be the trivial bound $2 \log_2 \dim C$). The chain rule mentioned above then implies that, on average over~$i$, we have ${I(A_i : C | B A_1 \dots A_{i-1})} \leq c/n$. The authors of~\cite{ILW08} gave evidence for the difficulty of characterising such states in the quantum setting by finding states for which the conditional mutual information is small whereas their distance to any Markov chain is large (see also~\cite{CSW12} for more extreme examples). Recent works by~\cite{WL, Kim13, Zha12} made the important observation that instead of considering the distance to a (perfect) Markov chain, another possibly more appropriate measure would be the accuracy with which Eq.~\ref{eq_markovchain} is satisfied. In fact, it was conjectured in~\cite{Kim13} that the conditional mutual information is lower bounded by the trace distance between the two sides of Eq.~\ref{eq_markovchain} for a specific form for the map $\cT_{B \to BC}$ known sometimes as the Petz map (cf.\ Eq.~\ref{eq_petzmap} below). Later, in the context of studying R\'enyi generalisations of the conditional mutual information, the authors of~\cite{BSW14} refined this conjecture by replacing the trace distance with the negative logarithm of the fidelity (see also~\cite{SesWil14}). Here, we prove a variant of this last conjecture where the map $\cT_{B \to BC}$ does not necessarily have the form of a Petz map.

\paragraph{Main result.} We prove that for any state $\rho_{A B C}$ on $A \otimes B \otimes C$ there exists a quantum operation $\mathcal{T}_{B \to B C}$ from the $B$ system to the $B \otimes C$ system such that the fidelity of the reconstructed state \begin{align} \label{eq_sigmadef}
  \sigma_{A B C} =  \mathcal{T}_{B \to B C}(\rho_{A B}) 
\end{align}
is at least\footnote{The \emph{fidelity} of $\rho$ and $\sigma$ is defined as $F(\rho , \sigma) = \| \sqrt{\rho} \sqrt{\sigma} \|_1$. }
\begin{align} \label{eq_maininequality}
   F(\rho_{A B C}, \sigma_{A B C}) 
  \geq 2^{-\frac{1}{2} I(A : C | B)_{\rho}}  \ .
\end{align}
We refer to Theorem \ref{thm_maininequality} for a more precise statement. 

\paragraph{Reformulations and implications.} A first immediate implication of our inequality is the \emph{strong subadditivity} of the von Neumann entropy,  ${I(A : C | B)}_{\rho} \geq 0$~\cite{LieRus73}. The latter may be rewritten in terms of the conditional von Neumann entropy, $H(A|B)_{\rho} = H(\rho_{A B}) - H(\rho_B)$, as  
\begin{align} \label{eq_strongsubadditivity}
  H(A|B)_\rho \geq H(A| B C)_\rho
\end{align}
and is also known as the \emph{data processing inequality}. Furthermore, \eqref{eq_maininequality}~implies that if~\eqref{eq_strongsubadditivity} holds with equality for some state $\rho_{A B C}$ then it satisfies the Markov chain condition~\eqref{eq_markovchain}, reproducing the result from~\cite{Pet88,HJPW04}. The work presented here may thus be viewed as a robust extension of this result | if~\eqref{eq_strongsubadditivity} holds with \emph{approximate} equality then the Markov chain condition is fulfilled \emph{approximately}. 

Our result may also be rewritten as 
\begin{align}
  \inf_{\sigma_{A B C}} D_{\frac{1}{2}}(\rho_{A B C}, \sigma_{A B C}) \leq I(A : C | B)_\rho \ ,
\end{align}
where the infimum ranges over all \emph{recovered} states, i.e., states of the form~\eqref{eq_sigmadef}, and where $D_{{1}/{2}}(\rho \| \sigma) = -2 \log_2 F(\rho, \sigma)$ is the R\'enyi divergence of  order $\alpha = {1}/{2}$~\cite{MDSFT13, WWY13}. We remark that the quantity on the left hand side is equal to the \emph{surprisal of the fidelity of recovery}, which has been introduced and studied in detail in~\cite{SesWil14}. 

Finally, we note that~\eqref{eq_maininequality} also implies an upper bound on the trace distance, which we denote by $\Delta(\cdot, \cdot)$, between $\rho_{A B C}$ and the recovered state $\sigma_{A B C}$, 
\begin{align} \label{eq_maininequalityc}
  \frac{1}{\ln 2} \Delta( \rho_{A B C}, \sigma_{A B C})^2 \leq I(A : C | B)_{\rho} \ .
\end{align}
The bound is readily verified using $\Delta(\cdot, \cdot)^2 \leq 1-F(\cdot, \cdot)^2$ (cf.\ Lemma~\ref{lem_tracedistancefidelity}) and $1-2^{-x} \leq \ln(2) x$.

\paragraph{Tightness.}  One may ask whether, conversely to our main result, the conditional mutual information of a state $\rho_{A B C}$ also gives a lower bound on its distance to any reconstructed state $\sigma_{A B C}$ of the form~\eqref{eq_sigmadef}. To answer this question, we note that, as a consequence of the data processing inequality, we have
     \begin{align}
     I(A : C | B)_{\rho}
     = H(A | B)_\rho - H(A | B C)_{\rho}
     \leq H(A | B C)_{\sigma} - H(A | B C)_\rho \ .
   \end{align}
  The entropy difference on the right hand side can be bounded by the Alicki-Fannes inequality~\cite{AF03} in terms of the trace distance between the two states, yielding\footnote{We refer to~\cite{BSW14} for a more detailed discussion, including a proof that the same bound holds also when the conditional mutual information is evaluated for $\sigma$ instead of $\rho$.} 
   \begin{align}
     I(A : C | B)_{\rho}
   \leq 8 \Delta \log_2(\dim A) - 4 \Delta \log_2(2 \Delta) - 2 (1-2 \Delta) \log_2(1-2 \Delta)   \qquad \text{for $\Delta \leq \frac{1}{2}$}\ .
   \end{align}
  This can be seen as a converse to~\eqref{eq_maininequalityc}. To simplify the comparison, we may use
    \begin{align}
     8 \Delta- 4 \Delta \log_2(2 \Delta) -  2 (1-2 \Delta) \log_2(1-2 \Delta) \leq 7 \sqrt{\Delta}  \qquad \text{for $\Delta \leq \frac{1}{11}$} \ ,
        \end{align}
     which gives
    \begin{align}
        I(A : C | B)_{\rho} \leq 7 \log_2(\dim A) \sqrt{\Delta(\rho_{A B C},   \sigma_{A B C})} \ . 
    \end{align}
Note that a term proportional to the logarithm of the dimension of $A$ is necessary in general as the trace distance is always upper bounded by~$1$, whereas the conditional mutual information may be as large as $2 \log_2 \dim A$. 

\paragraph{The classical case.}
  Inequality~\eqref{eq_maininequality} is easily obtained in the case where $B$ is classical, i.e., when $\rho_{A B C}$ is a qcq-state,
  \begin{align}
    \rho_{A B C} = \sum_{b} P_B(b) \, \proj{b}_B \otimes \rho_{AC,b} \ ,
  \end{align}
  for some probability distribution $P_{B}$, an orthonormal basis $\{\ket{b}\}_b$ of $B$, and a family of states $\{\rho_{A C, b}\}_{b}$ on $A \otimes C$.  Let $\mathcal{T}_{B \to B C}$ be any map such that
  \begin{align}
    \mathcal{T}_{B \to B C}(\proj{b}) = \proj{b} \otimes \rho_{C,b}  \qquad (\forall b) \ ,
  \end{align}
 where $\rho_{C, b} = \tr_A(\rho_{A C,b})$. Then the reconstructed state $\sigma_{A B C} = \mathcal{T}_{B \to B C}(\rho_{A B})$ is the qcq-state
\begin{align} \label{eq_qcqMarkov}
 \sigma_{A B C} = \sum_{b} P_B(b) \rho_{A, b} \otimes \proj{b} \otimes \rho_{C, b} \ ,
\end{align}
where $\rho_{A, b} = \tr_C(\rho_{A C, b})$.  We remark that $\sigma_{A B C}$ is a Markov chain. Furthermore, a straightforward calculation shows that the relative entropy\footnote{See Section~\ref{sec_relativeentropy} for a definition.} $D(\rho_{A B C} \| \sigma_{A B C})$ between  $\rho_{A B C}$ and $\sigma_{A B C}$ is given by
 \begin{align} \label{eq_classicalidentity}
   D(\rho_{A B C} \| \sigma_{A B C}) = I(A : C | B)_\rho \ .
 \end{align}
 Inequality~\eqref{eq_maininequality} then follows from Lemma~\ref{lem_DFidelity}. 
  
\paragraph{Related results.}  While the conditional mutual information is well understood in the classical case and has various interesting properties (see, e.g., \cite{R02}), these properties do not necessarily hold for quantum states. For example, identity~\eqref{eq_classicalidentity} cannot be generalised directly to the case where $B$ is non-classical (see~\cite{WL} for a discussion). Furthermore, it has been discovered that there exist states $\rho_{A B C}$ that have a large distance to the closest Markov chain, while the conditional mutual information is small~\cite{ILW08, CSW12, Erk14}. We remark that this is not in contradiction to~\eqref{eq_maininequality} as the reconstructed state $\sigma_{A B C}$, defined by~\eqref{eq_sigmadef}, is not necessarily a Markov chain. (Note that this is a major difference to the classical case sketched above.)

As mentioned above, the special case of~\eqref{eq_maininequality} where $I(A : C | B) = 0$ has been studied in earlier work~\cite{Pet88,HJPW04}. There, it has also been shown that the relevant reconstruction map $\mathcal{T}_{B \to B C}$ is of the form 
\begin{align}
\label{eq_petzmap}
  X_B \mapsto  {\rho_{B C}^{\frac{1}{2}} (\rho_B^{-\frac{1}{2}} X_B  \rho_B^{-\frac{1}{2}} \otimes \id_C) \rho_{B C}^{\frac{1}{2}}} \ .
\end{align}
However, it remained unclear whether this particular map also works in the case where $I(A : C | B)$ is strictly larger than zero, even though several conjectures in this direction were proposed and studied~\cite{WL,Kim13,Zha12,BSW14}. We refer to~\cite{LW14} for a detailed account of the evolution of these conjectures.
We note that our result provides some information about the structure of the map for which~\eqref{eq_maininequality} holds (cf.\ Theorem~\ref{thm_maininequality}), but leaves open the question whether it is of this particular form.  

There is a large body of literature underlying the fundamental role that the conditional mutual information plays in quantum information theory.  Notably, it has been shown to characterise the communication rate for the task of \emph{quantum state redistribution} in the asymptotic limit of many independent copies of a resource state~\cite{DY08}. Furthermore,  the quantum conditional mutual information is the basis for an important measure of entanglement, known as \emph{squashed entanglement}~\cite{ChrWin03}. The properties of this entanglement measure thus hinge on the properties of $I(A : C | B)$. In this context, lower bounds on $I(A : C | B)$ in terms of the distance between the marginal $\rho_{A C}$ from the set of separable states have been proved in~\cite{BCY11} and later improved in~\cite{LiWin14}.  We also note that another lower bound on the conditional mutual information in terms of a distance between certain operators derived from $\rho_{A B C}$ has recently been stated in~\cite{ZhaWu14}. This bound is based on a novel monotonicity bound for the relative entropy~\cite{CL14}. Our work may be used to obtain strengthened versions of some of these results. We are going to illustrate this for the case of squashed entanglement.

\paragraph{Applications.} 

For us, one motivation to study how well the conditional mutual information characterises approximate Markov chains is in the context of device-independent quantum key distribution~\cite{BDFR14}. Another implication, proposed in~\cite{WL,LW14},  is a novel lower bound on the squashed entanglement of any bipartite state. The bound depends only on the trace distance to the closest $k$-extendible\footnote{A non-negative operator $\omega_{A C}$ is called \emph{$k$-extendible} if there exists a non-negative operator $\bar{\omega}_{A C_1 \cdots C_n}$ such that $\bar{\omega}_{A C_i} = \omega_{A C}$ for all $i = 1, \ldots, k$. \label{ftn_extendible}} state, and also implies a strong lower bound in terms of the trace distance to the closest separable state (cf.\ Appendix~\ref{app_squashed} for details).

It would be interesting to investigate whether inequality~\eqref{eq_maininequality} can lead to better quantum de Finetti theorems. In fact, the authors of~\cite{BH13,BH13b} recently gave beautiful proofs of various de Finetti theorems using the conditional mutual information. For the quantum version, they apply an informationally complete measurement to reduce the problem to the classical case, but this comes at the cost of a factor that is exponential in the number of systems. We also believe that inequality \eqref{eq_maininequality} will be helpful in proving communication complexity lower bounds via the \emph{quantum information complexity}~\cite{JRS03,JN10,KLLRX12,Tou14}.

\paragraph{Structure of the proof.} The proof of inequality~\eqref{eq_maininequality} is based on two main ideas, which we discuss in separate sections. The first is the use of \emph{one-shot entropy measures}~\cite{Renner05} to bound the von Neumann relative entropy (Section~\ref{sec_relativeentropy}). The second is an extension of the method of \emph{de Finetti reductions}~\cite{Ren07,CKMR07,CKR09,Renner10} (Section~\ref{sec_deFinetti}). We use the latter to derive a general tool for evaluating the fidelity of permutation-invariant states (Section~\ref{sec_fidelity}). The proof of~\eqref{eq_maininequality} then proceeds in two main steps in which these techniques are applied successively (Section~\ref{sec_proof}). 
  
\section{Typicality bounds on the relative entropy} \label{sec_relativeentropy} 

In this section we are going to derive bounds on the relative entropy that will be used in the proof of Theorem~\ref{thm_maininequality}. The method we use to obtain these bounds is inspired by a recent approach~\cite{BeaRen12} to prove strong subadditivity of the von Neumann entropy (see Eq.~\ref{eq_strongsubadditivity}). The idea  there was to first prove strong subadditivity for one-shot entropies~\cite{Renner05} and then use \emph{typicality} or, more precisely, the \emph{Asymptotic Equipartition Property}~\cite{TCR09} to obtain the desired statement for the von Neumann entropy. Here we proceed analogously: we use one-shot versions of the relative entropy (defined in Appendix~\ref{app_relativeentropy}) to obtain bounds on the von Neumann relative entropy.

The \emph{(von Neumann) relative entropy} $D(\rho \| \sigma)$ for two non-negative operators  $\rho$ and $\sigma$ is defined as
\begin{align}
  D(\rho \| \sigma) = \frac{1}{\tr(\rho)}\tr\bigl(\rho(\log_2 \rho - \log_2 \sigma)\bigr) \ , 
\end{align}
where we set $D(\rho \| \sigma) = \infty$ if the support of $\rho$ is not contained in the support of $\sigma$.
Our statements also refer to the \emph{trace distance}. While  this distance is often defined for density operators only, we define it here more generally for any non-negative operators $\rho$ and $\sigma$ by
\begin{align}
  \Delta(\rho, \sigma) = \frac{1}{2} \|\rho - \sigma\|_1 + \frac{1}{2}\bigl|\tr(\rho - \sigma)\bigr| 
\end{align}
(see Section~3.2 of~\cite{Tom12}). Note that the second term is zero if $\rho$ and $\sigma$ are both density operators. We also remark that the trace distance may be rewritten as
\begin{align} \label{eq_tracedistancepositive}
  \Delta(\rho, \sigma) = \max\bigl[\tr(Y^+), \tr(Y^-)\bigr] \ ,
\end{align}
where $Y+$ and $Y^-$ are the positive and negative parts of $\rho - \sigma$, i.e., $\rho - \sigma = Y^+ - Y^-$ with $Y^+ \geq 0$, $Y^- \geq 0$, and $\tr(Y^+ Y^-) = 0$. It follows that we can write $\Delta$ as 
\begin{align} \label{eq_tracedistancemaximisation}
  \Delta(\rho, \sigma) =  \sup_{0 \leq Q \leq \id} |\tr(Q (\rho - \sigma))| \ .
\end{align}
\comment{ 
To prove this expression, we prove both inequalities. Taking two particular values for $Q$ being the projector onto the support of $Y^+$ and the projector on the support of $Y^{-}$, we see that $\Delta(\rho, \sigma) = \max\bigl[\tr(Y^+), \tr(Y^-)\bigr] \leq  \sup_{0 \leq Q \leq \id} |\tr(Q (\rho - \sigma))|$. Also for any $0 \leq Q \leq \id$, we have $\tr(Q(\rho - \sigma)) = \tr(Q(Y^+ - Y^{-})) \leq \tr(Q Y^+) \leq \tr(Y^+)$. Similarly, $\tr(Q(\rho - \sigma)) \geq - \tr(Q Y^{-}) \geq - \tr(Y^{-})$, which proves that $ \sup_{0 \leq Q \leq \id} |\tr(Q (\rho - \sigma))| \leq \Delta(\rho, \sigma)$. \\
}
One can easily see from this expression that for any trace non-increasing completely positive map $\cW$ we have
\begin{align} \label{eq_tracedistancemonotone}
  \Delta(\cW(\rho), \cW(\sigma)) \leq \Delta(\rho, \sigma)  \ .
\end{align}

\comment{
For any trace non-increasing completely positive map $\cW$,
\begin{align}
    \Delta(\cW(\rho), \cW(\sigma))
    = \sup_{0 \leq Q \leq \id}  |\tr(Q \cW(\rho - \sigma))|
    &=  \sup_{0 \leq Q \leq \id}  |\tr(\cW^*(Q)(\rho - \sigma))|  \\
    &\leq \sup_{0 \leq \bar{Q} \leq \id} |\tr(\bar{Q} (\rho - \sigma))|
    = \Delta(\rho, \sigma) \ ,
\end{align}
where $\cW^*$ is the adjoint map. The inequality holds because $\cW^*$ is sub-unital and completely positive, so that $0 \leq \cW^*(Q) \leq \id$.
}

Our first lemma provides an upper bound on the relative entropy in terms of sequences of operators that satisfy an operator inequality. 

\begin{lemma} \label{lem_boundfromsequence}
  Let  $\rho$ be a density operator, let $\sigma$ be a non-negative operator, and let $\{\bar{\rho}_n\}_{n \in \mathbb{N}}$ be a sequence of non-negative operators such that for some $s \in \mathbb{R}$
 \begin{align} \label{eq_limitdistancesmallerone}
   \bar{\rho}_n \leq 2^{s n} \sigma^{\otimes n} \quad  (\forall n \in \mathbb{N}) \qquad \text{and} \qquad \lim_{n \to \infty} \Delta(\bar{\rho}_n, \rho^{\otimes n})  < 1 \ .
  \end{align}
  Then $D(\rho \| \sigma) \leq s$.
\end{lemma}

\begin{proof}
  By assumption, there exist $c<1$ and $n_0 \in \mathbb{N}$ such that 
   \begin{align}
    \Delta(\bar{\rho}_n, \rho^{\otimes n}) \leq c
  \end{align}
  holds for all $n \geq n_0$.   Let $\epsilon \in (c,1)$. By Lemma~\ref{lem_DHupperbound} we have
    \begin{align}
    D_H^\epsilon(\rho^{\otimes n} \| \sigma^{\otimes n}) \leq  s n -\log_2 (1-c/\epsilon)   \ ,
    \end{align}
  where $D_H^\epsilon(\cdot \| \cdot)$ is the generalised relative entropy defined in Appendix~\ref{app_relativeentropy}. Setting $C = 1- c/\epsilon > 0$ we conclude that
  \begin{align}
    \lim_{n \to \infty} \frac{1}{n} D_H^\epsilon(\rho^{\otimes n} \| \sigma^{\otimes n}) 
  \leq \lim_{n \to \infty} \bigl(s + \frac{1}{n} \log_2 \frac{1}{C} \bigr) = s \ .
  \end{align}
  The claim then follows from the Asymptotic Equipartition Property of $D_H^\epsilon(\cdot \| \cdot)$ (Lemma~\ref{lem_QAEPequal}).
\end{proof}

The following lemma is in some sense a converse of Lemma~\ref{lem_boundfromsequence}. 

\begin{lemma} \label{lem_sequencefrombound}
  Let $\rho$ be a density operator, let $\sigma$ be non-negative operator, and let $s >     D(\rho \| \sigma)$.  Then there exists $\kappa > 0$ and a sequence of non-negative operators $\{\bar{\rho}_n\}_{n \in \mathbb{N}}$ with $\tr(\bar{\rho}_n) \leq 1$ such that  
  \begin{align} \label{eq_rhobarconditions}
     \bar{\rho}_n \leq 2^{s n} \sigma^{\otimes n}  \quad  (\forall n \in \mathbb{N})
 \qquad \text{and} \qquad  \lim_{n \to \infty} 2^{n \kappa} \Delta(\bar{\rho}_n, \rho^{\otimes n}) = 0 \ .
  \end{align}
\end{lemma}

\begin{proof}
The proof uses the smooth relative max-entropy $D_{\max}^\epsilon(\cdot \| \cdot)$ defined in Appendix~\ref{app_relativeentropy}. 
The Asymptotic Equipartition Property for this entropy measure (Lemma~\ref{lem_DAEP}) asserts that there exists $n_0 \in \mathbb{N}$ such that for any $n \geq n_0$
\begin{align}
  D^{\epsilon_n}_{\max}(\rho^{\otimes n} \| \sigma^{\otimes n}) < n s
\end{align}
for $\epsilon_n > 0$ chosen such that
  \begin{align} \label{eq_epsilonndef}
    D(\rho \| \sigma) + c \sqrt{\frac{\log_2(2/\eps_n^2)}{n}} = s \ ,
  \end{align}
  where $c$ is independent of~$n$. Inserting this into the definition of $D_{\max}^\epsilon(\cdot \| \cdot)$ we find that there exists a non-negative operator $\bar{\rho}_n$ with $\tr(\bar{\rho}_n) \leq 1$ such that 
  \begin{align}
    \bar{\rho}_n \leq 2^{s n} \sigma^{\otimes n}
  \end{align}
  and
  \begin{align} \label{eq_epsilonnbound}
     \sqrt{1 - F(\bar{\rho}_n, \rho^{\otimes n})^2} \leq \epsilon_n \ .
  \end{align}
   Eq.~\eqref{eq_epsilonndef} may be rewritten as
  \begin{align}
    \epsilon_n = \sqrt{2} \, 2^{-\kappa' n/2}
    \qquad \text{with} \quad
    \kappa' = \left(\frac{s-D(\rho\|\sigma)}{c}\right)^2 \ .
  \end{align}
  Inserting this in~\eqref{eq_epsilonnbound} and using Lemma~\ref{lem_tracedistancefidelity}, we conclude that
  \begin{align}
    \Delta(\bar{\rho}_n, \rho^{\otimes n}) \leq \sqrt{2} \, 2^{-\kappa' n/2} \ .
  \end{align}
  This proves~\eqref{eq_rhobarconditions} for any $\kappa < \kappa'/2$. (Note that for $n < n_0$ we may simply set $\bar{\rho}_n = 0$ so that the left hand side of~\eqref{eq_rhobarconditions}  holds for all $n \in \mathbb{N}$.)
  \end{proof}
  
  The next lemma asserts that the relative entropy, evaluated for $n$-fold product states, has the following stability property: if one acts with the same trace non-increasing map on the two arguments then the relative entropy cannot substantially increase. This property is used in the proof of Theorem~\ref{thm_maininequality}  (but see also Remark~\ref{rem_Dmap}). 
  
  \begin{lemma} \label{lem_Dmapping}
    Let $\rho$ be a density operator, let $\sigma$ be a non-negative operator on the same space, and let $\{\cW_n\}_{n \in \mathbb{N}}$ be a sequence of trace non-increasing completely positive maps on the $n$-fold tensor product of this space. If $\tr(\cW_n(\rho^{\otimes n}))$  decreases less than exponentially in $n$, i.e., 
    \begin{align}
      \liminf_{n \to \infty} e^{\xi n} \tr\bigl(\cW_n(\rho^{\otimes n})\bigr) > 0
    \end{align}
    for any $\xi > 0$, then
    \begin{align}
      \limsup_{n \to \infty} \frac{1}{n} D\bigl(\cW_n(\rho^{\otimes n}) \| \cW_n(\sigma^{\otimes n})\bigr) \leq D(\rho \| \sigma) \ .
    \end{align}
  \end{lemma}
  
  \begin{proof}
    Let $\delta > 0$. Lemma~\ref{lem_sequencefrombound} tells us that there exists $\kappa > 0$ and a sequence of non-negative operators $\{\bar{\rho}_m\}_{m \in \mathbb{N}}$ such  that
  \begin{align} \label{eq_CgivenABinequality}
    \bar{\rho}_m \leq 2^{m (D(\rho \| \sigma) +\delta)} \sigma^{\otimes m} 
  \end{align}
  and
  \begin{align}\label{eq_rhoABCndistanceboundstrong}
      \lim_{m \to \infty} e^{\kappa m} \Delta( \bar{\rho}_m, \rho^{\otimes m}) = 0 \ .
  \end{align}
  To abbreviate notation, we define  $r_n = 1/\tr(\cW_n(\rho^{\otimes n}))$.  Note that, by assumption, $r_n$ grows less than exponentially in $n$, so that  
  \begin{align} \label{eq_rnzero}
      r_n < e^{\kappa n}  
  \end{align}
 holds  for $n$ sufficiently large. 
    
  Let now $k,n \in \mathbb{N}$ and set $m = k n$. Applying $\cW_n$ and multiplying with the factor $r_n$ on the two sides of~\eqref{eq_CgivenABinequality} yields
  \begin{align}
   r_n^k \cW_n^{\otimes k}(\bar{\rho}_{n k})
  & \leq 2^{k n (D(\rho \| \sigma) + \delta)} \bigl(r_n \cW_n(\sigma^{\otimes n})\bigr)^{\otimes k}  \ .
  \end{align}
   As $\cW_n$ is trace non-increasing,  
    \begin{multline}\label{eq_rhoABCndistanceboundstrongp_2}
      \lim_{k \to \infty} \Delta\bigl(r_n^k \cW_n^{\otimes k}(\bar{\rho}_{n k}), \left(  r_n \cW_n(\rho^{\otimes n}) \right)^{\otimes k}\bigr)
      =     \lim_{k \to \infty} r_n^k \Delta\bigl(\cW_n^{\otimes k}(\bar{\rho}_{n k}),  \cW_n^{\otimes k}(\rho^{\otimes n k}) \bigr) \\
      \leq \lim_{k \to \infty} r_n^k \Delta( \bar{\rho}_{n k},  \rho^{\otimes n k} ) 
      \leq  \lim_{k \to \infty} e^{\kappa n k} \Delta( \bar{\rho}_{n k},  \rho^{\otimes n k} ) 
      = 0 \ ,
  \end{multline}
  where the  first inequality follows from the monotonicity property of the trace distance~\eqref{eq_tracedistancemonotone}, the second inequality follows from~\eqref{eq_rnzero}, and where the final equality follows from~\eqref{eq_rhoABCndistanceboundstrong}. We can now apply Lemma~\ref{lem_boundfromsequence} to the density operator $r_n \cW_n(\rho^{\otimes n})$ and the non-negative operator $r_n \cW_n(\sigma^{\otimes n})$, which gives
  \begin{align}
    D\bigl(r_n \cW_n(\rho^{\otimes n}) \| r_n \cW_n(\sigma^{\otimes n}) \bigr) \leq n\bigl(D(\rho \| \sigma) + \delta\bigr) \ .
  \end{align}
  Noting that multiplying both arguments of the relative entropy with the same factor leaves it unchanged we conclude
    \begin{align}
    \frac{1}{n} D\bigl(\cW_n(\rho^{\otimes n}) \| \cW_n(\sigma^{\otimes n}) \bigr) \leq D(\rho \| \sigma) + \delta \ .
  \end{align}  
  Taking the limit $n \to \infty$ and noting that $\delta > 0$ was arbitrary, the claim of the lemma follows. 
  \end{proof}

   \begin{remark} \label{rem_Dmap}
   Lemma~\ref{lem_Dmapping} will be used in one of the steps of the proof of Theorem~\ref{thm_maininequality}. We note that, alternatively, this step may also be based on the inequality (cf.\ Lemma~25 of~\cite{BCFST13})
  \begin{align} \label{eq_Dmonotone}
    (1-\epsilon) D\bigl(\cW(\rho) \| \cW(\sigma)\bigr) \leq D(\rho \| \sigma) + \epsilon \log_2(\tr(\sigma)/\epsilon) \ ,
  \end{align}
  which holds for any  density operator $\rho$, any non-negative operator $\sigma$,  any trace non-increasing completely positive map $\cW$, and $\epsilon = 1-\tr\bigl(\cW(\rho)\bigr)$ (see also Footnote~\ref{ftn_altproof}). However, Lemma~\ref{lem_Dmapping} provides a stronger stability condition for the relative entropy of product states (notably when $\epsilon \gg 0$) and may therefore be useful for generalisations of our results.
  
 \comment{\\ To prove~\eqref{eq_Dmonotone}, we define
  \begin{align}
   \bar{\cW} : \quad X \mapsto \cW(X) + \tr(X - \cW(X)) \omega
\end{align}
where $\omega$ is some density operator orthogonal to $\rho$ and $\sigma$. Because $\bar{\cW}$ is trace-preserving, we can use the monotonicity of the relative entropy under such maps to conclude that
  \begin{align}
    (1-\epsilon) D\bigl(\cW(\rho) \| \cW(\sigma)\bigr) + \epsilon D\bigl(\epsilon \omega \| \tr(\sigma - \cW(\sigma)) \omega\bigr)
    =  D\bigl(\bar{\cW}(\rho) \| \bar{\cW}(\sigma)\bigr) 
    \leq D(\rho \| \sigma) \ .
      \end{align}
  The inequality then follows from
   \begin{align}
      \epsilon D\bigl(\epsilon \omega \| \tr(\sigma - \cW(\sigma)) \omega\bigr)
    = \epsilon \log_2(\epsilon) - \epsilon \log_2(\tr(\sigma-\cW(\sigma)))
    \geq \epsilon \log_2(\epsilon) - \epsilon \log_2(\tr(\sigma))
   = - \epsilon \log_2(\tr(\sigma)/\epsilon) \ .
   \end{align}}
   \end{remark}
   
   As a corollary of Lemma~\ref{lem_Dmapping} we also obtain the well known \emph{Asymptotic Equipartition Property} (see, e.g., Chapter~3 of~\cite{CovTho05}).  We state it here explicitly as Lemma~\ref{lem_typicalsubspace} because we are going to use it within the proof of Theorem~\ref{thm_maininequality} and because it illustrates the use of Lemma~\ref{lem_Dmapping}. 
   
\begin{lemma} \label{lem_typicalsubspace} 
Let $\rho$ be a density operator. For any $n \in \mathbb{N}$ let $\rho^{\otimes n} = \sum_{s \in S_n} s \, \Pi_s$ where $S_n$ is the set of eigenvalues of $\rho^{\otimes n}$ and where $\Pi_s$, for $s \in S_n$, is the projector onto the corresponding eigenspace. Furthermore, for any $\delta > 0$, let $S_n^\delta$ be the subset of $S_n$ defined by
  \begin{align}
    S_n^\delta = \bigl\{s \in S_n : \, s \in [2^{-n(H(\rho) + \delta)}, 2^{-n(H(\rho) - \delta)}] \bigr\} \ .
  \end{align}
  Then $\lim_{n \to \infty} \sum_{s \in S_n^{\delta}} \tr(\Pi_s \rho^{\otimes n}) = 1$ and  the convergence is exponentially fast in $n$.
 \end{lemma}
 
 \begin{proof}
   Let $n \in \mathbb{N}$ and consider the projector $ \bar{\Pi}^+_n =\sum_{s \in S_n^{+}} \Pi_s$,  where $S_n^+ = \{s \in S_n : \, s  >  2^{-n(H(\rho)-\delta)}  \}$.  An explicit evaluation of the relative entropy shows that
   \begin{align}
     D(\bar{\Pi}^+_n \rho^{\otimes n} \bar{\Pi}^+_n\| \bar{\Pi}^+_n)  
     = \frac{\tr\bigl(\rho^{\otimes n} \bar{\Pi}_n^+ \log_2(\rho^{\otimes n}  \bar{\Pi}_n^+)\bigr)}{\tr(\rho^{\otimes n} \bar{\Pi}^+_n)}   
     \geq  \frac{ \tr\bigl(\rho^{\otimes n} \bar{\Pi}_n^+ \log_2(2^{-n(H(\rho)-\delta)} \bar{\Pi}_n^+)\bigr)  }{\tr(\rho^{\otimes n} \bar{\Pi}^+_n)}         
     =  - n \bigl(H(\rho) - \delta\bigr) \ . 
    \end{align}
  Using 
   \begin{align} \label{eq_entropyrelative}
     - H(\rho) = D(\rho \| \id)  
   \end{align}
   and defining the map $\cW^+_n : \, X \mapsto  \bar{\Pi}^+_n X  \bar{\Pi}^+_n$ we can rewrite this bound as
   \begin{align} \label{eq_Wndecrease}
     \frac{1}{n} D(\cW^+_n(\rho^{\otimes n}) \| \cW^+_n(\id^{\otimes n})) 
     \geq D(\rho \| \id) + \delta \ .
   \end{align}
   If we now assume, by contradiction, that  $\tr(\rho^{\otimes n} \bar{\Pi}^+_n) = \tr(\cW^+_n(\rho^{\otimes n}))$ decreases less than exponentially fast in~$n$, Lemma~\ref{lem_Dmapping} tells us that
   \begin{align}
     \limsup_{n \to \infty} \frac{1}{n} D(\cW^+_n(\rho^{\otimes n}) \| \cW^+_n(\id^{\otimes n})) \leq D(\rho \| \id)  \ .
   \end{align}
   This is obviously in contradiction to~\eqref{eq_Wndecrease} and thus proves that $\tr(\rho^{\otimes n} \bar{\Pi}^+_n)$ decreases exponentially fast in~$n$. 
   
   Similarly, we may consider the projector $ \bar{\Pi}^-_n =\sum_{s \in S_n^{-}} \Pi_s$ where $S_n^- = \{s \in S_n : \, s  <  2^{-n(H(\rho)+\delta)}  \}$.  Here, instead of~\eqref{eq_entropyrelative}, we use that for any purification $\rho_{D R}$ of $\rho_D = \rho$
   \begin{align} \label{eq_condentropyrelative}
     H(\rho) =  - H(D|R)_{\rho} = D(\rho_{D R} \| \id_D \otimes \rho_R) \ .
   \end{align} 
  We may choose the purification such that $(\bar{\Pi}^-_n \otimes \id_{R^n}) \rho_{D R}^{\otimes n} (\bar{\Pi}^-_n \otimes \id_{R^n})  = (\bar{\Pi}^-_n \otimes \bar{\Pi}^-_n) \rho_{D R}^{\otimes n} ( \bar{\Pi}^-_n \otimes \bar{\Pi}^-_n)$. Then 
   \begin{multline}
   D\bigl((\bar{\Pi}^-_n \otimes \id_{R^n}) \rho_{D R}^{\otimes n} (\bar{\Pi}^-_n \otimes \id_{R^n}) \big\| \bar{\Pi}^-_n \otimes \rho_R^{\otimes n}\bigr)
   = \log_2\tr(\rho_D^{\otimes n} \bar{\Pi}^-_n) - \frac{\tr\bigl( \rho_R^{\otimes n} \bar{\Pi}^-_n  \log_2(\rho_R^{\otimes n} \bar{\Pi}^-_n) \bigr)}{\tr(\rho_D^{\otimes n} \bar{\Pi}^-_n)} \\
   \geq  \log_2\tr(\rho_D^{\otimes n} \bar{\Pi}^-_n) - \frac{\tr\bigl( \rho_R^{\otimes n} \bar{\Pi}^-_n\log_2(2^{-n(H(\rho) + \delta)}\bar{\Pi}^-_n) \bigr)}{\tr(\rho_D^{\otimes n} \bar{\Pi}^-_n)} 
   =  \log_2\tr(\rho_D^{\otimes n} \bar{\Pi}^-_n) + n\bigl(H(\rho) + \delta\bigr) \ .
      \end{multline}
    Defining $\cW^-_n : \, X_{D R} \mapsto (\bar{\Pi}^-_n \otimes \id_R) X_{D R} (\bar{\Pi}^-_n \otimes \id_R)$ and inserting~\eqref{eq_condentropyrelative} we obtain the bound
   \begin{align} \label{eq_WDRbound}
    \frac{1}{n} D\bigl(\cW^-_n(\rho_{D R}^{\otimes n}) \big\| \cW^-_n (\id_{D}^{\otimes n} \otimes \rho_R^{\otimes n})\bigr)  \geq D(\rho_{D R} \| \id_D \otimes \rho_R) + \delta + \frac{\log_2 \tr\bigl(\cW^-_n(\rho_{D R}^{\otimes n})\bigr)}{n}  \ .
   \end{align}    
    Assume now, by contradiction, that $\tr(\rho_D^{\otimes n} \bar{\Pi}^-_n) = \tr(\cW^-_n(\rho_{D R}^{\otimes n}))$ decreases less then exponentially fast in~$n$. Then the last term of~\eqref{eq_WDRbound} approaches~$0$ in the limit of large~$n$. In particular, we have
   \begin{align}
     \limsup_{n \to \infty} \frac{1}{n} D\bigl(\cW^-_n(\rho_{D R}^{\otimes n}) \big\| \cW^-_n (\id_{D}^{\otimes n} \otimes \rho_R^{\otimes n})\bigr)  \geq D(\rho_{D R} \| \id_D \otimes \rho_R) + \delta \ ,
   \end{align}
 which contradicts the statement of  Lemma~\ref{lem_Dmapping}.  We have thus shown that both $\tr( \rho^{\otimes n} \bar{\Pi}^-_n)$ and  $\tr( \rho^{\otimes n} \bar{\Pi}^+_n )$ decrease exponentially fast in~$n$.  The claim of the lemma follows because $\sum_{s \in S_n^{\delta}}  \Pi_s = \id - \bar{\Pi}^+_n - \bar{\Pi}^-_n$.
 \end{proof}
 
 We conclude this section with a remark that is going to be useful for our proof of Theorem~\ref{thm_maininequality}. 
 
 \begin{remark} \label{rem_typicalsize}
  Considering the decomposition $\rho = \sum_{r \in R} r  \pi_r$, it is easy to see that all eigenvalues of $\rho^{\otimes n}$ have the form $\prod_{r \in R} r^{n_r}$, where  $(n_r)_{r \in R}$ are partitions of $n$, i.e., elements from the set 
  \begin{align}
  Q_n = \bigl\{(n_r)_{r \in R},  \, n_r \in \mathbb{N}_0, \, \sum_{r \in R} n_r = n \bigr\} \ .
\end{align}
Hence, the  set $S_n$ of eigenvalues of $\rho^{\otimes n}$ used within Lemma~\ref{lem_typicalsubspace} has size at most $|S_n| \leq |Q_n|$. Since $|Q_n| = \binom{n+|R|-1}{n} \leq (n+1)^{|R|}$, where $|R| \leq \mathrm{rank}(\rho)$ is the number of different eigenvalues of~$\rho$, we can upper bound the size of $S_n$ by
\begin{align}
  | S_n | \leq (n+1)^{\mathrm{rank}(\rho)}  \ .
\end{align}
\end{remark}

\section{Generalised de Finetti reduction} \label{sec_deFinetti}

The main result of this section, stated as Lemma~\ref{lem_postselection}, is motivated by a variant of the method of \emph{de Finetti reductions} proposed in~\cite{CKR09}. (This variant is also known as \emph{postselection technique}; we refer to~\cite{Renner10} for a not too technical presentation.) De Finetti reductions are generally used to study states on $n$-fold product systems $S^{\otimes n}$ that are invariant under permutations of the subsystems~\cite{Ren07,CKMR07}. More precisely, the idea is to reduce the analysis of any  density operator $\rho_{S^n}$ in the symmetric subspace $\Sym^n(S)$ of $S^{\otimes n}$ to the | generally simpler | analysis of states of the form $\sigma_{S}^{\otimes n}$, where $\sigma_S$ is pure. We extend this method to the case where $S = D \otimes E$ is a bipartite space and where the marginal of $\rho_{S^n} = \rho_{D^n E^n}$ on $D^{\otimes n}$ is known to have the form 
\begin{align} \label{eq_marginalcondition}
  \tr_{E^n}(\rho_{D^n E^n}) = \rho_{D^n} =  \sigma_{D}^{\otimes n} 
\end{align}
for some given state $\sigma_D$ on $D$. Lemma~\ref{lem_postselection} implies that, in this case, the analysis can be reduced to states of the form $\sigma_{D E}^{\otimes n}$, where $\sigma_{D E}$ is a purification of $\sigma_D$. (We note that a similar extension has been proposed earlier for another variant of the de Finetti reduction method; see Remark~4.3.3 of~\cite{Renner05}.)  Lemma~\ref{lem_postselection} will play a central role for the derivation of the claims of Section~\ref{sec_fidelity} below. Its proof uses  concepts from representation theory, which are presented in Appendix~\ref{app_representation}.

\begin{lemma} \label{lem_postselection}
  Let $D$ and $E$ be Hilbert spaces and let $\sigma_{D}$ be a non-negative operator on $D$. Then there exists a probability measure $\mathrm{d}\phi$ on the set of purifications $\proj{\phi}_{D E}$ of~$\sigma_{D}$ such that  
    \begin{align}
    \rho_{D^n E^n} \leq (n+1)^{d^2-1} \int \proj{\phi}_{D E}^{\otimes n} \mathrm{d} \phi 
  \end{align}
  holds for any $n \in \mathbb{N}$, any permutation-invariant purification $\rho_{D^n E^n}$ of $\sigma_D^{\otimes n}$, and $d = \max[\dim(D), \dim(E)]$. 
   \end{lemma}
 
  \begin{proof}
   For the following argument, we assume without loss of generality that  $d = \dim(D) = \dim(E)$, and that $\sigma_{D}$ has full rank and is therefore invertible on $D$. (If this is not the case one may embed the smaller space in one of dimension $d$ and replace $\sigma_D$ by $\sigma_D + \epsilon \,  \id_D$ for $\epsilon > 0$. The claim is then obtained in the limit $\epsilon \to 0$.) We define
 \begin{align}
   \ket{\theta}_{D E} = \sum_i \ket{d_i}_D \otimes \ket{e_i}_E \ ,
 \end{align}
 where $\{\ket{d_i}_D\}_i$ and $\{\ket{e_i}_E\}_i$ are orthonormal bases of $D$ and $E$, respectively. Let now
\begin{align}
  T_{D^n E^n} = \int (\id_{D^n} \otimes U_E^{\otimes n}) \proj{\theta}_{D E}^{\otimes n} (\id_{D^n} \otimes U_E^{\otimes n})^{\dagger} \mathrm{d} U \ ,
\end{align}
where $\mathrm{d} U$ is the Haar probability measure on the group of unitaries on $E$. Because  $(\sigma_D^{\frac{1}{2}} \otimes U_E) \proj{\theta} {(\sigma_D^{\frac{1}{2}} \otimes U_E^{\dagger})}$ is a purification of $\sigma_D$ for any $U_E$, the operator
\begin{align}
  \tau_{D^n E^n}
  =  (\sigma_D^{\otimes n} \otimes \id_{E^n})^{\frac{1}{2}} T_{D^n E^n} (\sigma_D^{\otimes n} \otimes \id_{E^n})^{\frac{1}{2}}
  =  \int \bigl((\sigma_D^{\frac{1}{2}} \otimes U_E) \proj{\theta}_{D E} (\sigma_D^{\frac{1}{2}} \otimes U_E^{\dagger})\bigr)^{\otimes n} \mathrm{d} U
\end{align}
is obviously of the form
\begin{align}
  \tau_{D^n E^n} = \int \proj{\phi}_{D E}^{\otimes n} \mathrm{d} \phi \ ,
\end{align}
for some suitably chosen measure $\mathrm{d} \phi$ on the set of purifications $\proj{\phi}_{D E}$ of $\sigma_{D}$. It therefore suffices to show that
\begin{align} \label{eq_tauform}
   \rho_{D^n E^n} \leq (n+1)^{d^2-1} \tau_{D^n E^n} \ .
\end{align}

We do this by analysing the structure of $T_{D^n E^n}$. For this we employ the Schur-Weyl duality, which equips the product space $(D \otimes E)^{\otimes n}$ with a convenient structure (see Appendix~\ref{app_representation}). Specifically, according to Lemma~\ref{lem_SchurWeyldecomposition}, the vector $\ket{\theta}_{D E}^{\otimes n}$ decomposes as
\begin{align} \label{eq_thetadecomposition}
  \ket{\theta}_{D E}^{\otimes n}
  =  \sum_{\lambda} \ket{\phi_\lambda}_{U_{D, \lambda} U_{E, \lambda}} \otimes \ket{\Psi_\lambda}_{V_{D, \lambda} V_{E,\lambda}} \ .
\end{align}
where, for each Young diagram $\lambda$,
\begin{align}
  \ket{\Psi_\lambda}_{V_{D, \lambda} V_{E,\lambda}}
  = \sqrt{\dim(V_\lambda)} \ket{\psi_\lambda}_{V_{D, \lambda} V_{E,\lambda}}
  = \sum_{k} \ket{v_k}_{V_{D, \lambda}} \otimes \ket{\bar{v}_k}_{V_{E, \lambda}}
  \end{align}
  for orthonormal bases $\{\ket{v_k}_{V_{D, \lambda}}\}_k$ and $\{\ket{\bar{v}_k}_{V_{E, \lambda}}\}_k$ of $V_{D,\lambda}$ and $V_{E, \lambda}$, respectively, and $ \ket{\phi_\lambda}_{U_{D, \lambda} U_{E, \lambda}}$ is a vector in $U_{D, \lambda} \otimes  U_{E, \lambda}$. The latter may always be written in the Schmidt decomposition as
 \begin{align} \label{eq_philambda}
   \ket{\phi_\lambda}_{U_{D, \lambda} U_{E, \lambda}} = \sum_{j} \alpha_{\lambda, j} \ket{u_j}_{U_{D, \lambda}} \otimes \ket{\bar{u}_j}_{U_{E, \lambda}} \ ,
 \end{align}
where $\{\ket{u_j}_{U_{D, \lambda}}\}_j$ and $\{\ket{\bar{u}_j}_{U_{E, \lambda}}\}_j$  are orthonormal bases of of $U_{D, \lambda}$ and $U_{E, \lambda}$, respectively, and $\alpha_{\lambda, j}$ are appropriately chosen coefficients, which we assume without loss of generality  to be real. The marginal of $\proj{\theta}_{D E}^{\otimes n}$ on  $D^{\otimes n} \cong \bigoplus_{D, \lambda} U_{D,\lambda} \otimes V_{D, \lambda}$ is equal to the identity and can thus be written as
\begin{align} \label{eq_marginalidentity}
  \tr_{E^n}(\proj{\theta}_{D E}^{\otimes n}) = \id_{D^n} = \sum_\lambda \id_{U_{D, \lambda}} \otimes \id_{V_{D, \lambda}} \ .
\end{align}
Comparing this to~\eqref{eq_thetadecomposition} shows that all coefficients $\alpha_{\lambda, j}$ in~\eqref{eq_philambda} must be equal to~$1$, i.e.,
 \begin{align}
   \ket{\phi_\lambda}_{U_{D, \lambda} U_{E, \lambda}} = \sum_{j} \ket{u_j}_{U_{D, \lambda}} \otimes \ket{\bar{u}_j}_{U_{E, \lambda}} \ .
 \end{align}
Note also that, according to the Schur-Weyl duality (see, e.g., Theorem~1.10 of~\cite{Christandl06}), $U_E^{\otimes n}$ acts on
 $E^{\otimes n} \cong \bigoplus_{E, \lambda} U_{E,\lambda} \otimes V_{E, \lambda}$ as $\sum_\lambda U_{E, \lambda}(U) \otimes \id_{V_{E, \lambda}}$. We thus have
 \begin{align}
     (\id_{D^n} \otimes U_E^{\otimes n}) \ket{\theta}^{\otimes n}
  =  \sum_{\lambda, j} \ket{u_j}_{U_{D, \lambda}} \otimes U_{E, \lambda}(U) \ket{\bar{u}_j}_{U_{E, \lambda}} \otimes \ket{\Psi_\lambda}_{V_{D, \lambda} V_{E, \lambda}} \ .
 \end{align}
 We may therefore write
 \begin{align} \label{eq_Tform}
   T_{D^n E^n} = \sum_{\lambda, \lambda', j, j'} \ket{u_j} \! \bra{u_{j'}}_{U_{D, \lambda} \gets U_{D, \lambda'}} \otimes (T_{\lambda, \lambda', j, j'})_{U_{E, \lambda} \gets U_{E, \lambda'}} \otimes \ket{\Psi_\lambda} \! \bra{\Psi_{\lambda'}}_{(V_{D, \lambda} V_{E, \lambda}) \gets (V_{D, \lambda'} V_{E, \lambda'})} \ ,
 \end{align}
 where $T_{\lambda, \lambda', j, j'}$ is the homomorphism between $U_{E, \lambda'}$ and $U_{E, \lambda}$ defined by
 \begin{align}
   (T_{\lambda, \lambda', j, j'})_{U_{E, \lambda} \gets U_{E, \lambda'}}  = \int U_{E, \lambda}(U) \ket{\bar{u}_j} \! \bra{\bar{u}_{j'}}_{U_{E, \lambda} \gets U_{E, \lambda'}} U_{E, \lambda'}(U)^{\dagger} \mathrm{d} U \ .
 \end{align}
 Since this operator manifestly commutes with the action of the unitary, Schur's lemma (see, e.g., Lemma~0.8 of~\cite{Christandl06}), together with the fact that $U_{E, \lambda}$ and $U_{E, \lambda'}$ are inequivalent for $\lambda \neq \lambda'$,  implies that it has the form
 \begin{align}
   (T_{\lambda, \lambda', j, j'})_{U_{E, \lambda} \gets U_{E, \lambda'}}  = \mu_{\lambda, j, j'} \delta_{\lambda, \lambda'} \id_{U_{E, \lambda}}
 \end{align}
for appropriately chosen coefficients $\mu_{\lambda, j, j'}$.  Inserting this in~\eqref{eq_Tform} gives
 \begin{align}
  T_{D^n E^n} = \sum_{\lambda, j, j'}  \mu_{\lambda, j, j'} \ket{u_j} \! \bra{u_{j'}}_{U_{D, \lambda}} \otimes \id_{U_{E, \lambda}} \otimes \proj{\Psi_\lambda}_{V_{D, \lambda} V_{E, \lambda}} \ .
 \end{align}
 Because the marginal  of $T_{D^n E^n}$ on $D^{\otimes n}$,
 \begin{align}
   T_{D^n} = \sum_{\lambda, j, j'}  \mu_{\lambda, j, j'} \dim(U_\lambda)  \ket{u_j} \! \bra{u_{j'}}_{U_{D, \lambda}} \otimes \id_{V_{D, \lambda}}  \ ,
 \end{align}
 must be equal to the marginal of $\proj{\theta}_{D E}^{\otimes n}$, we conclude from~\eqref{eq_marginalidentity} that $\mu_{\lambda, j, j'} = \frac{1}{\dim(U_\lambda)} \delta_{j, j'}$. Hence,
 \begin{align}
  T_{D^n E^n}
  = \sum_{\lambda} {\textstyle \frac{\dim(V_\lambda)}{\dim(U_\lambda)}} \,  \id_{U_{D, \lambda}} \otimes \id_{U_{E, \lambda}} \otimes \proj{\psi_\lambda}_{V_{D, \lambda} V_{E, \lambda}} \ ,
 \end{align}
 where $\ket{\psi_\lambda}_{V_{E, \lambda}}$ is normalised.

Defining the invertible operator
\begin{align}
  \kappa_{D^n} = \sum_{\lambda} {\textstyle \frac{\dim(V_{\lambda})}{\dim(U_\lambda)}} \, \id_{U_{D, \lambda}} \otimes \id_{V_{D, \lambda}}
\end{align}
we have
\begin{align}
  S_{D^n E^n}
  = (\kappa_{D^n} \otimes \id_{E^n})^{-\frac{1}{2}} T_{D^n E^n} (\kappa_{D^n} \otimes \id_{E^n})^{-\frac{1}{2}}
  =  \sum_{\lambda} \id_{U_{D, \lambda}} \otimes \id_{U_{E, \lambda}} \otimes \proj{\psi_\lambda}_{V_{D, \lambda} V_{E, \lambda}} \ .
\end{align}
Note that $\kappa_{D^n}$ commutes with any permutation, because, according to the Schur-Weyl duality, permutations act like  $\sum_\lambda \id_{U_\lambda} \otimes V_\lambda(\pi)$ on the decomposition of $D^{\otimes n}$. Consequently, because the support of $T_{D^n E^n}$ is contained in the symmetric subspace $\Sym^n(D \otimes E)$, the same must hold for $S_{D^n E^n}$.  Furthermore, for any vector  $\ket{\Omega} \in \Sym^n(D \otimes E)$, it follows from its representation according to Lemma~\ref{lem_SchurWeyldecomposition} that $S_{D^n E^n} \ket{\Omega} = \ket{\Omega}$. This proves that
\begin{align} \label{eq_Sidentity}
  S_{D^n E^n} = \id_{\Sym^n(D \otimes E)} \ .
\end{align}

Consider now the operator
\begin{align}
  Q_{D^n E^n} = (\kappa_{D^n} \otimes \id_{E^n})^{-\frac{1}{2}}  (\sigma_D^{\otimes n} \otimes \id_{E^n})^{-\frac{1}{2}} \rho_{D^n E^n}  (\sigma_D^{\otimes n} \otimes \id_{E^n})^{-\frac{1}{2}} (\kappa_{D^n} \otimes \id_{E^n})^{-\frac{1}{2}} \ .
\end{align}
Since the support of $\rho_{D^n E^n}$ is contained in $\Sym^n({D \otimes E})$, the same must hold for $Q_{D^n E^n}$ and we find
\begin{align}
  Q_{D^n E^n} \leq \|Q_{D^n E^n}\|_\infty  \id_{\Sym^n(D \otimes E)} \leq \tr(Q_{D^n E^n}) \id_{\Sym^n(D \otimes E)} =   \tr(Q_{D^n E^n}) S_{D^n E^n}  \ .
\end{align}
This, in turn, implies that
\begin{align}
  \rho_{D^n E^n} \leq \tr(Q_{D^n E^n}) (\sigma_D^{\otimes n} \otimes \id_{E^n})^{\frac{1}{2}} T_{D^n E^n} (\sigma_D^{\otimes n} \otimes \id_{E^n})^{\frac{1}{2}}
  = \tr(Q_{D^n E^n})  \tau_{D^n E^n} \ .
\end{align}
To conclude the proof of~\eqref{eq_tauform}, we note that $\rho_{D^n} = \sigma_D^{\otimes n} = \tau_{D^n}$, which implies
\begin{multline}
  \tr(Q_{D^n E^n})
= \tr\bigl((\kappa_{D^n} \otimes \id_{E^n})^{-1}   (\sigma_D^{\otimes n} \otimes \id_{E^n})^{-\frac{1}{2}} \rho_{D^n E^n}  (\sigma_D^{\otimes n} \otimes \id_{E^n})^{-\frac{1}{2}} \bigr) \\
= \tr\bigl(\kappa_{D^n}^{-1} (\sigma_D^{\otimes n})^{-\frac{1}{2}} \rho_{D^n} (\sigma_D^{\otimes n})^{-\frac{1}{2}}  \bigr) \\
= \tr\bigl(\kappa_{D^n}^{-1} (\sigma_D^{\otimes n})^{-\frac{1}{2}} \tau_{D^n} (\sigma_D^{\otimes n})^{-\frac{1}{2}}  \bigr) \\
= \tr\bigl((\kappa_{D^n} \otimes \id_{E^n})^{-1}   (\sigma_D^{\otimes n} \otimes \id_{E^n})^{-\frac{1}{2}} \tau_{D^n E^n}  (\sigma_D^{\otimes n} \otimes \id_{E^n})^{-\frac{1}{2}} \bigr) \\
= \tr\bigl((\kappa_{D^n} \otimes \id_{E^n})^{-1} T_{D^n E^n}\bigr) \\
= \tr(S_{D^n E^n})
= \tr(\id_{\Sym^n(D \otimes E)})
= \dim(\Sym^n(D \otimes E))
\leq (n+1)^{d^2-1} \ .
\end{multline}
 \end{proof}
 
Because any permutation-invariant density operator has a permutation-invariant purification,  Lemma~\ref{lem_postselection} can be easily extended so that $\rho_{D^n E^n}$ does not need to be pure.
 
 \begin{corollary} \label{cor_postselection}
  Let $D$ and $E$ be Hilbert spaces and let $\sigma_{D}$ be a non-negative operator on $D$. Then there exists a probability measure $\mathrm{d} \sigma_{D E}$ on the set of non-negative extensions $\sigma_{D E}$  of $\sigma_D$ such that  
      \begin{align} \label{eq_postselectionext}
    \rho_{D^n E^n} \leq (n+1)^{d^2-1} \int \sigma_{D E}^{\otimes n} \mathrm{d} \sigma_{D E} 
  \end{align}
  holds for any $n \in \mathbb{N}$, any permutation-invariant non-negative extension $\rho_{D^n E^n}$ of $\sigma_D^{\otimes n}$, and $d = \dim(D) \dim(E)^2$.
   \end{corollary}
   
   \begin{proof}
     According to Lemma~\ref{lem_symmetricpurification}, $\rho_{D^n E^n}$ has a permutation-invariant purification  $\rho_{D^n E^n R^n}$ with purifying system $R^{\otimes n}$, where $\dim R \leq \dim (D \otimes E)$.  Lemma~\ref{lem_postselection} with $E$ replaced by $E \otimes R$, applied to $\rho_{D^n E^n R^n}$, yields
   \begin{align}
\rho_{D^nE^nR^n} \leq (n+1)^{d^2 - 1} \int \proj{\phi}_{D E R}^{\otimes n} \mathrm{d} \phi  \ ,
\end{align}  
where $\mathrm{d} \phi$ is a probability measure on the purifications $\proj{\phi}_{D E R}$ of $\sigma_D$ and $d = \max[\dim D, \dim {E \otimes R}] \leq \dim D (\dim E)^2$.  Taking the partial trace over $R^{\otimes n}$ on both sides gives
   \begin{align}
     \rho_{D^nE^n} \leq (n+1)^{d^2 - 1} \int \tr_R(\proj{\phi}_{D E R})^{\otimes n} \mathrm{d} \phi  \ .
    \end{align}  
    The claim follows because the probability measure $\mathrm{d} \phi$ on the pure states $\proj{\phi}_{D E R}$ can be replaced by the induced measure $\mathrm{d} \sigma_{D E}$   on the marginal states $\sigma_{D E} =  \tr_R(\proj{\phi}_{D E R})$.  
   \end{proof}
 
Even though we do not need it here, it is worth pointing out that, by virtue of the Choi-Jamio\l{}kowski isomorphism~\cite{Jam72,Cho75}, the claim above can be rephrased in terms of completely positive trace-preserving maps.  As shown in~\cite{BHOS14}, this is useful to derive an improved variant of inequality~\eqref{eq_maininequality}.

 \begin{corollary} \label{cor_postselectionmap}
 Let $D$ and $E$ be Hilbert spaces. Then there exists a probability measure $\mathrm{d}\tau$ on the set of completely positive trace-preserving maps $\tau_{D \to E}$ such that\footnote{The inequality means that the difference between the right hand side and the left hand side is a completely positive map.}
  \begin{align}
  \label{eq_postselection_maps}
    \cW_{D^n \to E^n} \leq (n+1)^{d^2-1} \int \tau_{D \to E}^{\otimes n} \mathrm{d} \tau 
  \end{align}
  holds for any $n \in \mathbb{N}$, any completely positive trace-preserving map $\cW_{D^n \to E^n}$ that is permutation-invariant (i.e., $\cW \circ \pi = \pi \circ \cW$ for all permutations $\pi$), and $d = \dim(D) \dim(E)^2$.
\end{corollary}

\begin{proof}
Let $\rho_{D^nE^n} = J^{\otimes n}(\cW_{D^n \to E^n})$ where  $J$ denotes the Choi-Jamio\l{}kowski isomorphism on the mappings from~$D$ to~$E$. That is, $\rho_{D^nE^n} = ({\cW_{\bar{D}^n \to E^n} \otimes \cI_{D^n}})(\Psi_{\bar{D} D}^{\otimes n})$, where $\Psi_{\bar{D} D}$ is a maximally entangled state. The marginal of $\rho_{D^n E^n}$ on $D^n$ equals $\sigma_D^{\otimes n}$ with $\sigma_D = \frac{\id_D}{\dim D}$. Furthermore, as the map $\cW_{\bar{D}^n \to E^n} \otimes \cI_{D^n}$ is permutation invariant, so is the state $\rho_{D^n E^n}$. Hence \eqref{eq_postselectionext}~holds. Since the corresponding probability measure $\mathrm{d} \sigma_{D E}$ is restricted to the set of density operators $\sigma_{D E}$ with marginal $\sigma_D$, each $\sigma_{D E}$ is the image of a trace-preserving completely positive map $\tau_{D \to E}$ under the isomorphism $J$, i.e., $\sigma_{D E} = J(\tau_{D \to E})$. The claim thus follows by applying the inverse of $J^{\otimes n}$ to both sides of~\eqref{eq_postselectionext}.
\end{proof} 
 
 \section{Fidelity between permutation-invariant operators} \label{sec_fidelity}

The purpose of this section is to provide techniques to approximate the fidelity of permutation-invariant states. They play a key  role in the proof of Theorem~\ref{thm_maininequality}. The derivation of the statements below is based  on the generalised de Finetti reduction method introduced in Section~\ref{sec_deFinetti}. Furthermore, we will  use several established facts about the fidelity, which are summarised in Appendix~\ref{app_fidelity}.

\begin{lemma} \label{lem_fidelitysymmetricoptimal}
  Let $\rho_{D^n E^n}$ be a permutation-invariant non-negative operator on $(D \otimes E)^{\otimes n}$ and let $\sigma_D$ be a non-negative operator on $D$. Then there exists a non-negative extension $\sigma_{D E}$ of $\sigma_D$ on $D \otimes E$ such that 
  \begin{align}
   F(\rho_{D^n E^n}, \sigma_{D E}^{\otimes n})
  \geq  (n+1)^{-d^2/2} F(\rho_{D^n}, \sigma_{D}^{\otimes n}) \ ,
  \end{align}
  where $d = \dim(D) \dim(E)^2$.
  Furthermore, if $\rho_{D^n E^n}$ is pure then $\sigma_{D E}$ is pure and $d \leq \max[\dim(D), \dim(E)]$. 
\end{lemma}

\begin{proof}
  Let $\rho_{D^n E^n R^n} = \proj{\Psi}_{D^n E^n R^n}$ be a permutation-invariant purification of $\rho_{D^n E^n}$, i.e., $\ket{\Psi} \in {\Sym^n(D \otimes E \otimes R)}$, where $\dim(R) \leq \dim(D \otimes E)$. (That such a purification exists is the statement of Lemma~\ref{lem_symmetricpurification}. We also note that, if  $\rho_{D^n E^n}$ is already pure, then $R$ can be chosen to be the trivial space~$\mathbb{C}$, i.e., $\dim(R) = 1$.) According to Lemma~\ref{lem_fidelitysymmetricpurification}, there exists a permutation-invariant purification $\bar{\sigma}_{D^n E^n R^n}$ of $\sigma_{D}^{\otimes n}$ such that
  \begin{align} \label{eq_fidelitybarsigma}
    F(\rho_{D^n}, \sigma_D^{\otimes n})
    = F(\rho_{D^n E^n R^n}, \bar{\sigma}_{D^n E^n R^n}) 
    = \sqrt{\bra{\Psi} \bar{\sigma}_{D^n E^n R^n} \ket{\Psi} }\ .
  \end{align}
  
  Let $\Gamma$ be the set of vectors $\ket{\phi}_{D E R}$ on $D \otimes E \otimes R$ such that $\tr_{E R}(\proj{\phi}_{D E R}) = \sigma_D$. According to Lemma~\ref{lem_postselection} there exists a probability measure $\mathrm{d}\phi$  on $\Gamma$ such that
   \begin{align}
      \bar{\sigma}_{D^n E^n R^n} \leq (n+1)^{d^2-1}  \int  \proj{\phi}_{D E R}^{\otimes n} \mathrm{d}\phi \ ,
    \end{align}  
    where $d = \max[\dim(D), \dim(E) \dim(R)] \leq \dim(D) \dim(E)^2$.   Using this we find
  \begin{align} \label{eq_fidelitytau}
    (n+1)^{-d^2} \bra{\Psi} \bar{\sigma}_{D^n E^n R^n} \ket{\Psi}
   \leq  \int \bra{\Psi} (\proj{\phi}_{D E R}^{\otimes n}) \ket{\Psi} \mathrm{d}\phi
   \leq \max_{\ket{\phi} \in \Gamma} \bra{\Psi} (\proj{\phi}_{D E R}^{\otimes n}) \ket{\Psi} \ .
  \end{align}
  We now set $\sigma_{D E R} = \proj{\phi}_{D E R}$, where $\ket{\phi}_{D E R} \in \Gamma$ is a vector that maximises the above expression. Note that, by the definition of the set $\Gamma$, $\sigma_{D E}$ is then a valid extension of the given operator $\sigma_D$.  (Furthermore, if $R$ is the trivial space $\mathbb{C}$ then $\sigma_{D E}$ is pure.) Combining~\eqref{eq_fidelitybarsigma} with~\eqref{eq_fidelitytau}, and using the monotonicity of the fidelity under the partial trace (Lemma~\ref{lem_fidelitypartialtrace}),  we conclude that
  \begin{align}
      (n+1)^{-d^2/2}  F(\rho_{D^n}, \sigma_D^{\otimes n})
   \leq \sqrt{\bra{\Psi} \sigma_{D E R}^{\otimes n} \ket{\Psi}}
   = F(\proj{\Psi}_{D^n E^n R^n}, \sigma_{D E R}^{\otimes n})
   \leq F(\rho_{D^n E^n}, \sigma_{D E}^{\otimes n})  \ .
  \end{align}  
\end{proof}

\begin{lemma} \label{lem_subfidelity}
  Let $\rho_{R^n S^n}$ be a permutation-invariant non-negative operator on $(R \otimes S)^{\otimes n}$ and let $\sigma_{R S}$ be a non-negative operator on $R \otimes S$. Furthermore, let $W_{R^n}$ be a permutation-invariant operator on $R^{\otimes n}$ with $\|W_{R^n}\|_\infty \leq 1$. Then there exists a unitary $U_R$  on $R$ such that\footnote{Here and in the following we simplify our notation by omitting identity operators that are clear from the context, e.g., we write $U_R$ instead of $U_R \otimes \id_S$.}
  \begin{align}  \label{eq_subfidelity}
    F\bigl(\rho_{R^n S^n}, U_R^{\otimes n}  \sigma_{R S}^{\otimes n} (U_R^{\otimes n})^{\dagger}\bigr) \geq (n+1)^{-d^2} F\bigl(W_{R^n} \rho_{R^n S^n} W_{R^n}^{\dagger},  \sigma_{RS}^{\otimes n} \bigr) \ ,
  \end{align}
  where $d =  \dim(R) \dim(S)^2$.
\end{lemma}
\begin{proof}
Let $\rho_{R^n S^n E^n}$ be a permutation-invariant purification of $\rho_{R^n S^n}$ on $(R \otimes S \otimes E)^{\otimes n}$, where $\dim(E) = \dim(R \otimes S)$ (cf.\ Lemma~\ref{lem_symmetricpurification}). Then $W_{R^n} \rho_{R^n S^n E^n} W_{R^n}^{\dagger}$ is a permutation-invariant purification of  $W_{R^n} \rho_{R^n S^n} W_{R^n}^{\dagger}$. Hence, according to Lemma~\ref{lem_fidelitysymmetricoptimal}, there exists a purification $\sigma_{R S E}$ of $\sigma_{R S}$ such that
\begin{align} \label{eq_Fpureprojector}
   (n+1)^{d_1^2/2}  F(W_{R^n} \rho_{R^n S^n E^n} W^{\dagger}_{R^n},  \sigma_{R S E}^{\otimes n} )
  \geq
  F(W_{R^n} \rho_{R^n S^n} W^{\dagger}_{R^n}, \sigma_{R S}^{\otimes n} ) \ ,
\end{align}
where $d_1 = \max[\dim(R \otimes S), \dim(E)] = \dim(R) \dim(S)$. We then use Lemma~\ref{lem_fidelityreduced} which asserts that
\begin{align} \label{eq_reductionstep}
  F(\rho_{S^n E^n}, \sigma_{S E}^{\otimes n})  
  \geq
  F(W_{R^n} \rho_{R^n S^n E^n} W^{\dagger}_{R^n}, \sigma_{R S E}^{\otimes n} ) \ .
\end{align}
Furthermore, again by Lemma~\ref{lem_fidelitysymmetricoptimal}, there exists a purification $\tilde{\sigma}_{R S E}$ of $\sigma_{S E}$ such that
\begin{align} \label{eq_purificationstep}
   (n+1)^{d_2^2/2} F(\rho_{R^n S^n E^n}, \tilde{\sigma}_{R S E}^{\otimes n})
  \geq  F(\rho_{S^n E^n}, \sigma_{S E}^{\otimes n}) \ ,
\end{align}
where $d_2 = \max[\dim(S \otimes E), \dim(R)] = \dim(R) \dim(S)^2$. Because all purifications are unitarily equivalent, there exists a unitary $U_R$ on $R$ such that $U_R \sigma_{R S E} U_R^\dagger = \tilde{\sigma}_{R S E}$, that is,
\begin{align}
  F\bigl(\rho_{R^n S^n E^n}, U_R^{\otimes n} \sigma_{R S E}^{\otimes n} (U_R^{\otimes n})^{\dagger}\bigr) = F(\rho_{R^n S^n E^n}, \tilde{\sigma}_{R S E}^{\otimes n}) \ .
\end{align}
Because the fidelity is non-decreasing under the partial trace (cf.\ Lemma~\ref{lem_fidelitypartialtrace}), we also have
\begin{align} \label{eq_Fpureiid}
  F\bigl(\rho_{R^n S^n}, U_R^{\otimes n}  \sigma_{R S}^{\otimes n} (U_R^{\otimes n})^{\dagger}\bigr)
  \geq F\bigl(\rho_{R^n S^n E^n}, U_R^{\otimes n}  \sigma_{R S E}^{\otimes n} (U_R^{\otimes n})^{\dagger}\bigr) \ .
\end{align}
Combining all these equations, we obtain the desired claim.
\end{proof}

\begin{remark} \label{rem_Udiagonal}
  If for some orthonormal basis $\{\ket{r}\}_r$ of $R$ the operator $W_{R^n}$ is diagonal in the corresponding product basis $\{\ket{r_1} \otimes \cdots \otimes \ket{r_n}\}_{r_1, \ldots r_n}$ then inequality~\eqref{eq_subfidelity} also holds for an operator $U_R$ which is diagonal in the basis $\{\ket{r}\}_r$ and satisfies $\| U_R \|_\infty \leq 1$, and for $d =  \dim(R)^2 \dim(S)^2$.
  
  To see this, let $\bar{R}$ be a system that is isomorphic to $R$ and let $C$ be the isometry from $R$ to $R \otimes \bar{R}$ defined by
  \begin{align}
     C = \sum_r \bigl(\ket{r}_R \otimes \ket{r}_{\bar{R}}\bigr)\bra{r}_R \ .
  \end{align}
  It is straightforward to verify that, for $W_{R^n}$ diagonal in the product basis $\{\ket{r_1} \otimes \cdots \otimes \ket{r_n}\}_{r_1, \ldots r_n}$, we have
  \begin{align}
    W_{R^n} = (C^{\dagger})^{\otimes n} (W_{R^n} \otimes \id_{\bar{R}^n}) C^{\otimes n} \ .
  \end{align}
  Let now $\rho_{R^n S^n E^n}$ and $\sigma_{R S E}$ be pure operators such that~\eqref{eq_Fpureprojector} holds. Furthermore, define $\bar{\rho}_{R^n \bar{R}^n S^n E^n} = C^{\otimes n}  \rho_{R^n S^n E^n} (C^{\dagger})^{\otimes n}$ and $\bar{\sigma}_{R \bar{R} S E} = C \sigma_{R S E} C^{\dagger}$. Using Lemma~\ref{lem_fidelityoperator} we find
  \begin{multline} \label{eq_Wsigmabar}
    F(W_{R^n} \rho_{R^n S^n E^n} W_{R^n}^{\dagger}, \sigma_{R S E}^{\otimes n})
  =
     F\bigl((C^{\dagger})^{\otimes n} W_{R^n} C^{\otimes n} \rho_{R^n S^n E^n} (C^{\dagger})^{\otimes n} W_{R^n}^{\dagger} C^{\otimes n}, \sigma_{R S E}^{\otimes n}\bigr) \\
   = F(W_{R^n} \bar{\rho}_{R^n \bar{R}^n S^n E^n} W_{R^n}^{\dagger},  \bar{\sigma}_{R \bar{R} S E}^{\otimes n}) \ .
  \end{multline}
  Furthermore, we can carry out the proof steps as in~\eqref{eq_reductionstep} and~\eqref{eq_purificationstep}, while keeping the system $\bar{R}^{\otimes n}$, to obtain
  \begin{align} \label{eq_purestepbar}
     (n+1)^{d^2/2} F(\bar{\rho}_{R^n \bar{R}^n S^n E^n}, \tilde{\sigma}_{R \bar{R} S E}^{\otimes n})
    \geq F(\bar{\rho}_{\bar{R}^n S^n E^n}, \bar{\sigma}_{\bar{R} S E}^{\otimes n} ) 
    \geq F(W_{R^n} \bar{\rho}_{R^n \bar{R}^n S^n E^n} W_{R^n}^{\dagger},  \bar{\sigma}_{R \bar{R} S E}^{\otimes n})  \ ,
  \end{align}
  for some purification $\tilde{\sigma}_{R \bar{R} S E}$ of $\bar{\sigma}_{\bar{R} S E}$. Because $\bar{\sigma}_{R \bar{R} S E}$ is pure there must exist a unitary $\bar{U}_R$ on $R$ such that $\bar{U}_R \bar{\sigma}_{R \bar{R} S E} \bar{U}_R^\dagger = \tilde{\sigma}_{R \bar{R} S E}$.   Using this and the fact that the fidelity is non-decreasing under the partial trace we find
  \begin{align} \label{eq_fidelityndbar}
  F\bigl(\bar{\rho}_{R^n \bar{R}^n S^n}, \bar{U}_R^{\otimes n}  \bar{\sigma}_{R \bar{R} S}^{\otimes n} (\bar{U}_R^\dagger)^{\otimes n} \bigr)
  \geq   F\bigl(\bar{\rho}_{R^n \bar{R}^n S^n E^n}, \bar{U}_R^{\otimes n}  \bar{\sigma}_{R \bar{R} S E}^{\otimes n} (\bar{U}_R^{\dagger})^{\otimes n}\bigr)
  = F\bigl(\bar{\rho}_{R^n  \bar{R}^n S^n E^n},   \tilde{\sigma}_{R \bar{R} S E}^{\otimes n} \bigr) \ .
  \end{align}
  Finally, by the definition of $\bar{\rho}_{R^n \bar{R}^n S^n E^n}$ and $\bar{\sigma}_{R \bar{R} S E}$, and using Lemma~\ref{lem_fidelityoperator}  we have
  \begin{align} \label{eq_UbarU}
    F\bigl(\bar{\rho}_{R^n \bar{R}^n S^n}, \bar{U}_R^{\otimes n}  \bar{\sigma}_{R \bar{R} S}^{\otimes n} (\bar{U}_R^\dagger)^{\otimes n}\bigr)
  = F\bigl(\rho_{R^n S^n}, U_R^{\otimes n} \sigma_{R S}^{\otimes n} (U_R^{\dagger})^{\otimes n} \bigr)
  \end{align}
  where $U_R = C^{\dagger} (\bar{U}_R \otimes \id_{\bar{R}}) C$. Combining this with~\eqref{eq_Fpureprojector}, \eqref{eq_Wsigmabar}, \eqref{eq_purestepbar}, and~\eqref{eq_fidelityndbar} we obtain again inequality~\eqref{eq_subfidelity}.  Furthermore, by construction, $U_R$ is diagonal in the basis $\{\ket{r}\}_r$ and satisfies $U^{\dagger}_R U_R \leq \id_R$. 
\end{remark}

\begin{remark} \label{rem_subfidelity}
  If $\rho_{R^n S^n}$ has product form $\rho_{R S}^{\otimes n}$, the statement of Lemma~\ref{lem_subfidelity} can be rewritten as
  \begin{align}
    F(\rho_{R S}, U_R \sigma_{R S} U_R^{\dagger})
    \geq \sqrt[n]{(n+1)^{-d^2} F(W_{R^n} \rho_{R S}^{\otimes n} W^{\dagger}_{R^n}, \sigma_{R S}^{\otimes n} )} \ .
  \end{align}
  Hence, for a family $\{W_{R^n}\}_{n \in \mathbb{N}}$ of permutation-invariant non-negative operators  such that $\|W_{R^n}\|_\infty \leq 1$ we have
  \begin{align} 
    \sup_{U_R} F\bigl(\rho_{R S}, U_R \sigma_{R S} U_R^{\dagger}\bigr) \geq \limsup_{n \to \infty} \sqrt[n]{F(W_{R^n} \rho_{R S}^{\otimes n} W^{\dagger}_{R^n}, \sigma_{R S}^{\otimes n} )} \ .
  \end{align}
\end{remark}

\section{Main result and proof} \label{sec_proof}

\begin{theorem} \label{thm_maininequality}
  For any density operator $\rho_{A B C}$ on $A \otimes B \otimes C$, where $A$, $B$, and $C$ are separable Hilbert spaces, there exists a trace-preserving completely positive map $\mathcal{T}_{B \to B C}$ from the space of operators on $B$ to the space of operators on $B \otimes C$ such that\footnote{$\cI_A$~denotes the identity map on the space of operators on $A$. We include it here in our notation to stress that the map leaves the $A$ system unaffected, but will usually omit it when it is clear from the context.} 
\begin{align} \label{eq_maininequalityp}
    2^{-\frac{1}{2} I(A : C | B)_\rho}
\leq F\bigl(\rho_{A B C}, (\cI_A \otimes \mathcal{T}_{B \to B C})(\rho_{A B})  \bigr) \ .
\end{align}
Furthermore, if $A$, $B$, and $C$  are finite-dimensional then $\mathcal{T}_{B \to B C}$ has the form
\begin{align} \label{eq_mappingform}
  X_B \mapsto V_{B C} \rho_{B C}^{\frac{1}{2}} (\rho_B^{-\frac{1}{2}} U_B  X_B U_B^{\dagger} \rho_B^{-\frac{1}{2}} \otimes \id_C) \rho_{B C}^{\frac{1}{2}} V_{B C}^{\dagger} 
\end{align}
on the support of $\rho_B$, where $U_B$ and $V_{B C}$ are unitaries on $B$ and $B \otimes C$, respectively.
\end{theorem}

\begin{proof}
We first note that, by Remark~\ref{rem_infinitedimensional} below, it is sufficient to prove the statement for the case where $A$, $B$, and $C$ are finite-dimensional. Let $\delta > 0$, $\delta' > 0$, and $\delta'' > 0$. Let $n \in \mathbb{N}$ and let $\{\Pi_b\}_{b \in \bar{B}_n}$ and $\{\Pi_d\}_{d \in \bar{D}_n}$ be the families of projectors onto the eigenspaces of $\rho_B^{\otimes n}$ and $\rho_{B C}^{\otimes n}$, labelled by their eigenvalues $b \in \bar{B}_n$ and $d \in \bar{D}_n$, respectively. Furthermore, let $\bar{B}_n^{\delta'}$ and $\bar{D}_n^{\delta''}$ be the subsets of $\bar{B}_n$ and $\bar{D}_n$ defined by Lemma~\ref{lem_typicalsubspace} and define
\begin{align}
  \Pi_{B^n}  = \sum_{b \in \bar{B}_n^{\delta'}} \Pi_b \qquad \text{and} \qquad  \Pi_{B^n C^n}  = \sum_{d \in \bar{D}_n^{\delta''}} \Pi_d \ .
\end{align}
Note that for any $\eta > 0$ we have
\begin{align}  \label{eq_projtypical}
  \tr(\Pi_{B^n} \rho_B^{\otimes n})  \geq 1-\eta \qquad \text{and} \qquad
    \tr(\Pi_{B^n C^n} \rho_{B C}^{\otimes n})  \geq 1 - \eta 
\end{align}
for $n$ sufficiently large. Define the mapping on $(A \otimes B \otimes C)^{\otimes n}$
\begin{align}
  \cW_n : \quad X_{A^n B^n C^n} \mapsto (\id_{A^n} \otimes \Pi_{B^n C^n}) (\id_{A^n} \otimes \Pi_{B^n} \otimes \id_{C^n}) X_{A^n B^n C^n} (\id_{A^n} \otimes \Pi_{B^n} \otimes \id_{C^n}) (\id_{A^n} \otimes \Pi_{B^n C^n}) \ .
\end{align}
as well as the abbreviation
\begin{align}
   \Gamma_{A^n B^n C^n} =  \cW_n(\rho_{A B C}^{\otimes n}) = \Pi_{B^n C^n} \Pi_{B^n} \rho_{A B C}^{\otimes n} \Pi_{B^n} \Pi_{B^n C^n} \ .
\end{align}
It is easily seen that the map $\cW_n$ is trace non-increasing and completely positive. Furthermore, because of~\eqref{eq_projtypical}, we always have $\tr(\cW_n(\rho_{A B C}^{\otimes n})) = \tr(\Gamma_{A^n B^n C^n})  > 2/3$ for $\eta$ sufficiently small (using the gentle measurement lemma, see~\cite{Wil11} for instance). Lemma~\ref{lem_Dmapping} then tells us that, for $n$ sufficiently large,\footnote{We note that a similar conclusion may be obtained from the inequality given in Remark~\ref{rem_Dmap}. \label{ftn_altproof}}
\begin{multline}
    D\bigl(\Gamma_{A^n B^n C^n} \| \Pi_{B^n C^n} \Pi_{B^n} \rho_{A B}^{\otimes n} \Pi_{B^n} \Pi_{B^n C^n} \bigr)  
    = D\bigl( \cW_n(\rho_{A B C}^{\otimes n}) \big\| \cW_n((\rho_{A B} \otimes \id_C)^{\otimes n}) \bigr)  \\
    \leq n \bigl(D(\rho_{A B C} \| \rho_{A  B} \otimes \id) + \frac{\delta}{2}\bigr) 
    = n (-H(C|AB) + \frac{\delta}{2})  \ ,
  \end{multline}
where the last equality is the definition of the conditional entropy, $H(C|AB) =  - {\tr(\rho_{A B C} \log_2 \rho_{A B C})} + {\tr(\rho_{A B} \log_2 \rho_{A B})} =  - {D(\rho_{A B C} \| \rho_{A B} \otimes \id_C)}
$.  The relation between the fidelity and the relative entropy (Lemma~\ref{lem_DFidelity}) now allows us to  conclude that
\begin{align} 
  \frac{1}{\tr(\Gamma_{A^n B^n C^n})} F(\Gamma_{A^n B^n C^n}, \Pi_{B^n C^n} \Pi_{B^n} \rho_{A B}^{\otimes n} \Pi_{B^n} \Pi_{B^n C^n}) \geq 2^{\frac{1}{2} n (H(C|AB)-\frac{\delta}{2})} \ . 
\end{align}
We now use Lemma~\ref{lem_fidelityoperator} to remove the projector $\Pi_{B^n C^n}$ from the second argument and note that the factor $\tr(\Gamma_{A^n B^n C^n}) > 2/3$ can be absorbed by another factor $2^{-\frac{1}{4} n \delta}$ for $n$ sufficiently large. This shows that
\begin{align} \label{eq_typical}
  F(\Gamma_{A^n B^n C^n}, \Pi_{B^n} \rho_{A B}^{\otimes n} \Pi_{B^n}) \geq 2^{\frac{1}{2} n (H(C|AB)-\delta)} \ .
\end{align}

Because $\sum_{b \in \bar{B}_n} \Pi_b = \id_{B^n}$ we can  apply Lemma~\ref{lem_fidelitydecomposition}, which gives
\begin{multline}
  F(\Gamma_{A^n B^n C^n},  \Pi_{B^n} \rho_{A B}^{\otimes n} \Pi_{B^n}) 
  \leq \sum_{b \in \bar{B}_n}  F(\Gamma_{A^n B^n C^n},  \Pi_b \Pi_{B^n} \rho_{A B}^{\otimes n} \Pi_{B^n} \Pi_b) \\
    = \sum_{b \in \bar{B}_n^{\delta'}}  F(\Gamma_{A^n B^n C^n},  \Pi_b \rho_{A B}^{\otimes n}  \Pi_b)  
    \leq  |\bar{B}_n^{\delta'}| \max_{b \in \bar{B}_n^{\delta'}}  F(\Gamma_{A^n B^n C^n},  \Pi_b \rho_{A B}^{\otimes n}  \Pi_b)  
 \ ,
\end{multline}
where the equality follows from
\begin{align}
  \Pi_b \Pi_{B^n} = \begin{cases} \Pi_b & \text{if $b \in \bar{B}_n^{\delta'}$} \\ 0 & \text{otherwise.} \end{cases} 
\end{align}
Hence, there exists $b \in \bar{B}_n^{\delta'}$ such that
\begin{align} \label{eq_cctwo}
  F(\Gamma_{A^n B^n C^n},  \Pi_{B^n} \rho_{A B}^{\otimes n} \Pi_{B^n}) 
\leq |\bar{B}_n^{\delta'}| F(\Gamma_{A^n B^n C^n},  \Pi_b \rho_{A B}^{\otimes n}  \Pi_b) 
\leq \mathrm{poly}(n) F(\Gamma_{A^n B^n C^n},  \Pi_b \rho_{A B}^{\otimes n}  \Pi_b) \ ,
\end{align}
where the second inequality follows from Remark~\ref{rem_typicalsize}. By the definition of $\Pi_b$ we also have
\begin{align}
  \Pi_b =  \sqrt{b}  (\rho_B^{-\frac{1}{2}})^{\otimes n} \Pi_b \ ,
\end{align}
where $b$ is the eigenvalue of $\rho_B^{\otimes n}$ corresponding to $\Pi_b$. By the definition of $\bar{B}_n^{\delta'}$ we also have $\sqrt{b} \leq 2^{- \frac{1}{2} n(H(B) - \delta')}$ and, hence, 
\begin{multline} \label{eq_ccthree}
  F(\Gamma_{A^n B^n C^n},  \Pi_b \rho_{A B}^{\otimes n}  \Pi_b) 
  = \sqrt{b} F\bigl(\Gamma_{A^n B^n C^n},  (\rho_B^{-\frac{1}{2}})^{\otimes n}  \Pi_b \rho_{A B}^{\otimes n}  \Pi_b (\rho_B^{-\frac{1}{2}})^{\otimes n} \bigr) \\
  \leq  2^{-\frac{1}{2} n (H(B) - \delta')} F\bigl(\Gamma_{A^n B^n C^n},  (\rho_B^{-\frac{1}{2}})^{\otimes n}  \Pi_b \rho_{A B}^{\otimes n}  \Pi_b (\rho_B^{-\frac{1}{2}})^{\otimes n} \bigr) \\
   =   2^{-\frac{1}{2} n (H(B) - \delta')} F\bigl(\Pi_b (\rho_B^{-\frac{1}{2}})^{\otimes n}  \Gamma_{A^n B^n C^n} (\rho_B^{-\frac{1}{2}})^{\otimes n} \Pi_b,  \rho_{A B}^{\otimes n}  \bigr) \ ,
\end{multline}
where the equality follows from Lemma~\ref{lem_fidelityoperator}, which we will use repeatedly in the following. Furthermore, by Lemma~\ref{lem_subfidelity}, there must exist a unitary $U_B$ on $B$ such that
\begin{multline} \label{eq_ccfour}
  F\big(\Pi_b (\rho_B^{-\frac{1}{2}})^{\otimes n} \Gamma_{A^n B^n C^n}  (\rho_B^{-\frac{1}{2}})^{\otimes n} \Pi_b,  \rho_{A B}^{\otimes n} \bigr) 
    \leq  \mathrm{poly}(n)  F\bigl( (\rho_B^{-\frac{1}{2}})^{\otimes n} \Gamma_{A^n B^n C^n} (\rho_B^{-\frac{1}{2}})^{\otimes n} ,  U_B^{\otimes n} \rho_{A B}^{\otimes n}   (U_B^{\otimes n})^{\dagger}\bigr)  \\
    =   \mathrm{poly}(n)  F\bigl( \Gamma_{A^n B^n C^n}, (\rho_B^{-\frac{1}{2}})^{\otimes n} U_B^{\otimes n} \rho_{A B}^{\otimes n}   (U_B^{\otimes n})^{\dagger}  (\rho_B^{-\frac{1}{2}})^{\otimes n} \bigr)  \ .
\end{multline}
Combining now~\eqref{eq_typical}, \eqref{eq_cctwo}, \eqref{eq_ccthree}, and~\eqref{eq_ccfour} we obtain
\begin{multline} \label{eq_ccpart}
  2^{\frac{1}{2} n (H(C|A B) + H(B) - \delta - \delta')}
\leq  \mathrm{poly}(n) F\bigl(\Gamma_{A^n B^n C^n},  (\rho_B^{-\frac{1}{2}})^{\otimes n} U_B^{\otimes n}  \rho_{A B}^{\otimes n}  (U_B^{\otimes n})^{\dagger} (\rho_B^{-\frac{1}{2}})^{\otimes n} )\bigr) \\
=  \mathrm{poly}(n) F(\Pi_{B^n C^n} \Pi_{B^n}  \rho_{A B C}^{\otimes n} \Pi_{B^n} \Pi_{B^n C^n}, \gamma_{A B C}^{\otimes n}) \ ,
\end{multline}
where $\gamma_{A B C} =  \rho_B^{-\frac{1}{2}} U_B  \rho_{A B} U_B^{\dagger} \rho_B^{-\frac{1}{2}}$.

Next we use that $\sum_{d \in \bar{D}_n} \Pi_d = \id_{B^n C^n}$ and apply again Lemma~\ref{lem_fidelitydecomposition} to obtain
\begin{multline}
  F(\Pi_{B^n C^n} \Pi_{B^n}  \rho_{A B C}^{\otimes n} \Pi_{B^n} \Pi_{B^n C^n},  \gamma_{A B C}^{\otimes n}) 
  \leq \sum_{d \in \bar{D}_n}  F(\Pi_d \Pi_{B^n C^n} \Pi_{B^n}  \rho_{A B C}^{\otimes n} \Pi_{B^n} \Pi_{B^n C^n} \Pi_d, \gamma_{A B C}^{\otimes n}) \\
    = \sum_{d \in \bar{D}_n^{\delta''}}  F(\Pi_d \Pi_{B^n}  \rho_{A B C}^{\otimes n} \Pi_{B^n} \Pi_d, \gamma_{A B C}^{\otimes n}) 
    \leq |\bar{D}_n^{\delta''}| \max_{d \in \bar{D}_n^{\delta''}} F(\Pi_d \Pi_{B^n}  \rho_{A B C}^{\otimes n} \Pi_{B^n} \Pi_d, \gamma_{A B C}^{\otimes n})  \ ,
\end{multline}
where $|\bar{D}_n^{\delta''}| \leq \mathrm{poly}(n)$ by Remark~\ref{rem_typicalsize}. Hence, there exists $d \in \bar{D}_n^{\delta''}$ such that
\begin{align} \label{eq_cdone}
  F(\Pi_{B^n C^n} \Pi_{B^n}  \rho_{A B C}^{\otimes n} \Pi_{B^n} \Pi_{B^n C^n},  \gamma_{A B C}^{\otimes n}) 
\leq \mathrm{poly}(n) F(\Pi_d \Pi_{B^n}  \rho_{A B C}^{\otimes n} \Pi_{B^n} \Pi_d, \gamma_{A B C}^{\otimes n}) \ .
\end{align}
By the definition of $\bar{D}_n^{\delta''}$ we have
\begin{align}
  \Pi_d = \frac{1}{\sqrt{d}} (\rho_{B C}^{\frac{1}{2}})^{\otimes n} \Pi_d
  \end{align}
with $d \geq 2^{-n(H(B C) +  \delta'')}$. This implies 
\begin{multline} \label{eq_cdtwo}
  F(\Pi_d \Pi_{B^n}  \rho_{A B C}^{\otimes n} \Pi_{B^n} \Pi_d, \gamma_{A B C}^{\otimes n})
=  \sqrt{\frac{1}{d}}  F\bigl((\rho_{B C}^{\frac{1}{2}})^{\otimes n} \Pi_d \Pi_{B^n}  \rho_{A B C}^{\otimes n} \Pi_{B^n} \Pi_d (\rho_{B C}^{\frac{1}{2}})^{\otimes n}, \gamma_{A B C}^{\otimes n}\bigr) \\
\leq 2^{\frac{1}{2} n (H(B C) + \delta'')} F\bigl((\rho_{B C}^{\frac{1}{2}})^{\otimes n} \Pi_d \Pi_{B^n}  \rho_{A B C}^{\otimes n} \Pi_{B^n} \Pi_d (\rho_{B C}^{\frac{1}{2}})^{\otimes n}, \gamma_{A B C}^{\otimes n} \bigr) \\
= 2^{\frac{1}{2} n (H(B C) + \delta'')} F\bigl(\Pi_d \Pi_{B^n}  \rho_{A B C}^{\otimes n} \Pi_{B^n} \Pi_d, (\rho_{B C}^{\frac{1}{2}})^{\otimes n} \gamma_{A B C}^{\otimes n} (\rho_{B C}^{\frac{1}{2}})^{\otimes n}\bigr) \ .
\end{multline}
We use again Lemma~\ref{lem_subfidelity}, which asserts that there must exist a unitary $V_{B C}$ on $B \otimes C$ such that
\begin{align} \label{eq_VBCstep}
  F\bigl(\Pi_d \Pi_{B^n}  \rho_{A B C}^{\otimes n} \Pi_{B^n} \Pi_d, (\rho_{B C}^{\frac{1}{2}})^{\otimes n} \gamma_{A B C}^{\otimes n} (\rho_{B C}^{\frac{1}{2}})^{\otimes n}\bigr) 
  \leq \mathrm{poly}(n) F\bigl(\rho_{A B C}^{\otimes n}, V_{B C}^{\otimes n} (\rho_{B C}^{\frac{1}{2}})^{\otimes n} \gamma_{A B C}^{\otimes n} (\rho_{B C}^{\frac{1}{2}})^{\otimes n} (V_{B C}^{\otimes n})^{\dagger}\bigr) \ .
\end{align}
Combining this with~\eqref{eq_ccpart}, \eqref{eq_cdone} and~\eqref{eq_cdtwo}  yields
\begin{align}
  2^{\frac{1}{2} n (H(C | A B) + H(B) - H(B C) - \delta - \delta' - \delta'')} 
  \leq \mathrm{poly}(n) F\bigl(\rho_{A B C}^{\otimes n}, V_{B C}^{\otimes n} (\rho_{B C}^{\frac{1}{2}})^{\otimes n} \gamma_{A B C}^{\otimes n} (\rho_{B C}^{\frac{1}{2}})^{\otimes n} (V_{B C}^{\otimes n})^{\dagger}\bigr) \ .
\end{align}
We take the $n$th root, use $H(B C) - H(B) - H(C | A B) = I(A: C | B)$, and insert the expression for $\gamma_{A B C}$ to rewrite this as
\begin{multline}
    2^{-\frac{1}{2} I(A : C | B) - \delta - \delta' - \delta''}
\leq  \sqrt[n]{\mathrm{poly}(n)} F(\rho_{A B C},  V_{B C} \rho_{B C}^{\frac{1}{2}} \rho_B^{-\frac{1}{2}} U_B  \rho_{A B} U_B^{\dagger} \rho_B^{-\frac{1}{2}} \rho_{B C}^{\frac{1}{2}} V_{B C}^{\dagger} ) \\
 \leq  \sqrt[n]{\mathrm{poly}(n)} \max_{U_B, V_{BC}} F(\rho_{A B C},  V_{B C} \rho_{B C}^{\frac{1}{2}} \rho_B^{-\frac{1}{2}} U_B  \rho_{A B} U_B^{\dagger} \rho_B^{-\frac{1}{2}} \rho_{B C}^{\frac{1}{2}} V_{B C}^{\dagger} ) \ ,
\end{multline}
where the maximum is take over unitary transformations $U_B$ and $V_{BC}$. As this maximised fidelity is now independent of $n$ and because $\sqrt[n]{\mathrm{poly}(n)}$ approaches $1$ for $n$ large, we conclude that
\begin{align}
    2^{-\frac{1}{2} I(A : C | B) - \delta - \delta' - \delta''}
\leq \max_{U_B, V_{BC}} F(\rho_{A B C},  V_{B C} \rho_{B C}^{\frac{1}{2}} \rho_B^{-\frac{1}{2}} U_B  \rho_{A B} U_B^{\dagger} \rho_B^{-\frac{1}{2}} \rho_{B C}^{\frac{1}{2}} V_{B C}^{\dagger} ) \ .
\end{align}
Inequality~\eqref{eq_maininequalityp} now follows because $\delta > 0$, $\delta' > 0$, and $\delta'' > 0$ were arbitrary.    

It remains to verify that the map $\mathcal{T}_{B \to B C}$ is trace-preserving. But this  follows from the observation that
\begin{align} \label{eq_Ttracenonincreasing}
  \tr_C(U_B^{\dagger} \rho_B^{-\frac{1}{2}} \rho_{B C}^{\frac{1}{2}} V_{B C}^{\dagger} V_{B C} \rho_{B C}^{\frac{1}{2}} \rho_B^{-\frac{1}{2}} U_B)
  =  \tr_C(U_B^{\dagger} \rho_B^{-\frac{1}{2}} \rho_{B C} \rho_B^{-\frac{1}{2}} U_B)
  = U_B^{\dagger} \rho_B^{-\frac{1}{2}} \rho_{B} \rho_B^{-\frac{1}{2}} U_B
  = \id_B \ .
\end{align}
\end{proof}

\begin{remark} \label{rem_infinitedimensional}
Any proof of the main claim of Theorem~\ref{thm_maininequality}, 
  \begin{align}  \label{eq_maininequalitypp}
    2^{-\frac{1}{2} I(A : C | B)_\rho}
\leq \sup_{\mathcal{T}_{B \to B C}} F\bigl(\rho_{A B C}, (\cI_A \otimes \mathcal{T}_{B \to B C})(\rho_{A B})  \bigr) \ ,
  \end{align}
which uses the assumption that $A$, $B$, and $C$ are finite-dimensional Hilbert spaces, implies that the claim also holds under the less restrictive assumption that these spaces are separable.
  
To see this, let $\{P_A^{k}\}_{k \in \mathbb{N}}$, $\{P_B^{k}\}_{k \in \mathbb{N}}$, and $\{P_C^{k}\}_{k_C \in \mathbb{N}}$ be  sequences of finite-rank projectors on $A$, $B$, and $C$ which converge to $\id_A$, $\id_B$, and $\id_C$, respectively, with respect to the weak (and, hence, also the strong) operator topology (see, e.g., Definition~2 of~\cite{FAR11}). Define furthermore the density operators
\begin{align}
  \rho_{A B C}^{k, k'} = \frac{(P_A^k \otimes P_{B}^{k'} \otimes P_{C}^k) \rho_{A B C} (P_A^k \otimes P_{B}^{k'} \otimes P_{C}^k)}{\tr\bigl((P_A^k \otimes P_{B}^{k'} \otimes P_{C}^k) \rho_{A B C}\bigr)}
\end{align}
and
\begin{align}
  \rho_{A B C}^k = \frac{(P_A^k \otimes \id_B \otimes P_{C}^k) \rho_{A B C} (P_A^k \otimes \id_B \otimes P_{C}^k)}{\tr\bigl((P_A^k \otimes \id_B \otimes P_{C}^k) \rho_{A B C}\bigr)} \ .
\end{align}
We note that, for any $k \in \mathbb{N}$, the sequence $\{\rho_{A B C}^{k, k'}\}_{k' \in \mathbb{N}}$  converges to $\rho_{A B C}^k$ in the trace-norm (see, e.g., Corollary~2 of~\cite{FAR11}).  Also, $\{\rho_{A B C}^k\}_{k \in \mathbb{N}}$ converges to $\rho_{A B C}$ in the trace norm. 

Let us first consider the left hand side of~\eqref{eq_maininequalitypp}. Because, for any fixed finite dimension of system $A$, the conditional mutual information ${I(A : C | B)}_\rho = H(A|B)_\rho - H(A | B C)_\rho$ is continuous in $\rho$ with respect to the trace norm~\cite{AF03}, we have
\begin{align} \label{eq_Ikpcont}
  \lim_{k' \to \infty} I(A : C | B)_{\rho^{k, k'}} = I(A : C | B)_{\rho^{k}} 
\end{align}
for any $k \in \mathbb{N}$. In addition, using the fact that local projectors applied to the subsystems $A$ and $C$ can only decrease $I(A: C | B)_\rho$, provided we scale by the probability of such a projector,
\begin{align}
  \tr\bigl((P_A^k \otimes \id_B \otimes P_{C}^k) \rho_{A B C}\bigr) I(A : C | B)_{\rho^{k}} \leq I(A : C | B)_{\rho} 
\end{align}
holds for any $k \in \mathbb{N}$. 
\comment{To see this, consider first a single local projector, say $P^k_A$ on $A$, in which case the claim reads 
\begin{align*}
  I(A : C | B)_{\rho} \geq \tr(P^k_A \rho) I(A : C | B)_{P^k_A \rho P^k_A / \tr(P^k_A \rho)} \ .
\end{align*}
To prove this, assume that a measurement with respect to $P^k_A$ as well as its orthogonal complement is applied to $\rho$. Let furthermore $Z$ be a random variable that stores the outcome of this measurement. Then, by strong subadditivity, we have 
\begin{align*}
  I(A : C | B)_\rho = H(C|B)_\rho - H(C|A B)_\rho \geq H(C|B)_{\rho'} - H(C|ABZ)_{\rho'} \geq H(C|B Z)_{\rho'} - H(C|ABZ)_{\rho'} = I(A : C | B Z)_{\rho'} \ ,
\end{align*}
where $\rho'$ denotes the state after the measurement. Because $ I(A : C | B Z)_{\rho'}$ can be written as the expectation over the mutual information of the post-measurement states conditioned on the different values of $Z$, and because all these terms are non-negative, the above claim follows.   By symmetry, the same argument can be repeated for a projector $P^k_C$ on $C$. }
Because $\lim_{k \to \infty} \tr\bigl((P_A^k \otimes \id_B \otimes P_{C}^k) \rho_{A B C}\bigr) = \tr(\rho) = 1$, we find
\begin{align}
  \limsup_{k \to \infty} I(A : C | B)_{\rho^{k}}  \leq I(A : C | B)_{\rho} \ .
\end{align}
The combination of this statement with~\eqref{eq_Ikpcont} yields
\begin{align} \label{eq_Iasymptotics}
  2^{-\frac{1}{2} I(A : C | B)_\rho} \leq \liminf_{k \to \infty} \lim_{k' \to \infty} 2^{-\frac{1}{2} I(A : C | B)_{\rho^{k, k'}}} \ .
\end{align}

We now consider the right hand side of~\eqref{eq_maininequalitypp}. Let $\delta > 0$ and note that, for sufficiently large $k$ and $k'$, we have 
\begin{align}
  \bigl\| \rho_{A B C}^{k, k'} - \rho_{A B C} \bigr\|_1 < (\delta/2)^2 \ .
\end{align}
Because the trace norm is monotonically non-increasing under trace-preserving completely positive maps, we also have
\begin{align}
  \bigl\| \mathcal{T}_{B \to B C}(\rho_{A B}^{k, k'}) - \mathcal{T}_{B \to B C}(\rho_{A B}) \bigr\|_1 < (\delta/2)^2
\end{align} 
for any $\mathcal{T}_{B \to B C}$. Lemma~\ref{lem_fidelitycontinuoustracenorm} then implies that 
\begin{align}
   F\bigl(\rho_{A B C}^{k, k'}, \mathcal{T}_{B \to B C}(\rho_{A B}^{k, k'}) \bigr)
  < F\bigl(\rho_{A B C}, \mathcal{T}_{B \to B C}(\rho_{A B})  \bigr)  + \delta
\end{align}
But because this holds for any $\mathcal{T}_{B \to B C}$, we have
\begin{align}
  \sup_{\mathcal{T}_{B \to B C}} F\bigl(\rho_{A B C}^{k, k'}, \mathcal{T}_{B \to B C}(\rho_{A B}^{k, k'}) \bigr)
  \leq \sup_{\mathcal{T}_{B \to B C}} F\bigl(\rho_{A B C}, \mathcal{T}_{B \to B C}(\rho_{A B})  \bigr)  + \delta \ .
\end{align}
Because this holds for all $\delta > 0$ and sufficiently large $k$ and $k'$, we find that
\begin{align} \label{eq_Fasymptotics}
  \limsup_{k \to \infty} \limsup_{k' \to \infty}   \sup_{\mathcal{T}_{B \to B C}} F\bigl(\rho_{A B C}^{k, k'},  \mathcal{T}_{B \to B C}(\rho_{A B}^{k, k'}) \bigr)
  \leq \sup_{\mathcal{T}_{B \to B C}} F\bigl(\rho_{A B C}, \mathcal{T}_{B \to B C}(\rho_{A B})  \bigr)  \ .
\end{align}

To conclude the argument, we observe that if the inequality~\eqref{eq_maininequalitypp} is valid for finite-dimensional spaces $A$, $B$, and $C$ we have in particular
\begin{align}
  \liminf_{k \to \infty} \lim_{k' \to \infty} 2^{-\frac{1}{2} I(A : C | B)_{\rho^{k, k'}}}
\leq   \limsup_{k \to \infty} \limsup_{k' \to \infty}   \sup_{\mathcal{T}_{B \to B C}} F\bigl(\rho_{A B C}^{k, k'}, \mathcal{T}_{B \to B C}(\rho_{A B}^{k, k'}) \bigr) \ .
\end{align}
Combining this with~\eqref{eq_Iasymptotics} and~\eqref{eq_Fasymptotics} then proves the claim that the inequality holds for arbitrary separable spaces  $A$, $B$, and $C$. 
\end{remark}

\begin{remark}
  By Remark~\ref{rem_Udiagonal}, the unitary $U_B$ chosen in~\eqref{eq_ccfour} may be replaced by an operator which commutes with $\rho_B$ and satisfies $\|U_B\|_{\infty} \leq 1$. Analogously, the unitary $V_{B C}$ chosen in~\eqref{eq_VBCstep} may be replaced by an operator of the form $V_{B C} = V'_B V''_{B C}$ where $V'_B$ commutes with~$\rho_B$ and $V''_{B C}$ commutes with~$\rho_{B C}$, and where $\|V'_B\|_\infty \leq 1$ and $\|V''_{B C}\|_\infty \leq 1$. Similarly to~\eqref{eq_Ttracenonincreasing} one can see that the resulting recovery map $\cT_{B \to B C}$ is trace non-increasing. Furthermore, we have
  \begin{multline}
     \cT_{B \to B C}(\rho_B) 
     = V_{B C} \rho_{B C}^{\frac{1}{2}} (\rho_B^{-\frac{1}{2}} U_B  \rho_B U_B^{\dagger} \rho_B^{-\frac{1}{2}} \otimes \id_C) \rho_{B C}^{\frac{1}{2}} V_{B C}^{\dagger} 
     = V_{B C} \rho_{B C}^{\frac{1}{2}} (U_B  U_B^{\dagger} \otimes \id_C) \rho_{B C}^{\frac{1}{2}}  V_{B C}^{\dagger}  \\
     \leq  V_{B C} \rho_{B C}^{\frac{1}{2}} \rho_{B C}^{\frac{1}{2}} V_{B C}^{\dagger}   
     = V'_B  \rho_{B C}^{\frac{1}{2}} V''_{B C} (V''_{B C})^{\dagger}  \rho_{B C}^{\frac{1}{2}} (V'_B)^{\dagger}
    \leq V'_B  \rho_{B C} (V'_B)^{\dagger} \ .
  \end{multline}
  In particular, we have
  \begin{align}
    \tr_C\bigl(\cT_{B \to B C}(\rho_B)\bigr) \leq \rho_B \quad \text{and} \quad \tr_B\bigl(\cT_{B \to B C} (\rho_B)\bigr) \leq \rho_C \ .
  \end{align}
  This implies that one can always choose a recovery map that exactly reproduces the marginal on $B$ and the marginal on~$C$. 
  
  \comment{One may specify this recovery map explicitly as follows. Let $\rho'_{B C} = \cT_{B \to B C}(\rho_B)$, where $\cT_{B \to B C}$ is the trace non-increasing map defined above. Furthermore, define $\omega_B = \rho_B - \rho'_B$ and $\omega_C = \rho_C - \rho'_C$. Note that $\omega_B$ and $\omega_C$ are non-negative and have the same trace $\tr(\omega_B) = \tr(\omega_C) = 1 - \tr(\rho'_{B C})$. A trace-preserving recovery map may then be defined by
  \begin{align} \label{eq_recoverymapadded}
    X_B \mapsto \cT(X_B) + \bigl(\tr(X_B) - \tr\cT(X_B)\bigr) \frac{\omega_B}{\tr(\omega_B)} \otimes \frac{\omega_C}{\tr(\omega_C)} \ .
  \end{align}
  Note that the second term is non-negative for any non-negative input. The map is therefore completely positive. Furthermore, for $X_B = \rho_B$ we have
  \begin{align}
    \tr(X_B) - \tr\cT(X_B)
    = \tr(\rho_B) - \tr(\rho'_{B C}) = \tr(\omega_B) = \tr(\omega_C)  \ .
   \end{align}
  Hence, the output of the map defined by~\eqref{eq_recoverymapadded} on input $\rho_B$ is
  \begin{align}
    \rho'_{B C} + \frac{\omega_B \otimes \omega_C}{\tr(\omega_C)} \ .
  \end{align} 
  The marginal of this operator on $B$ is obviously equal to $\rho'_B = \omega_B = \rho_B$.  Similarly, the marginal on $C$ is equal to $\rho_C$. }
 \end{remark}
        
\appendix

\section*{Appendices}

\section{One-shot relative entropies} \label{app_relativeentropy}

In this appendix, we briefly review the \emph{generalised relative entropy} introduced in~\cite{WanRen12,DKFRR13} and the \emph{smooth max-relative entropy} introduced in~\cite{Datta08} (we will use a slightly modified variant defined in~\cite{TCR09,Tom12}). 

\begin{definition} \label{def_oneshotentropy}
For any two non-negative operators $\rho$ and $\sigma$ and for any $\epsilon \in [0,\tr(\rho)]$, the \emph{generalised relative entropy}  is defined  by\footnote{For $\epsilon = 0$ the quantity is defined via continuous extension, i.e., $D_H^0(\rho \| \sigma) = \lim_{\epsilon \downarrow 0} D_H^\epsilon(\rho \| \sigma)$.}
\begin{align} \label{eq_DHdefinition}
  2^{- D_H^\epsilon(\rho \| \sigma)} = \inf_{\substack{0 \leq Q \leq 1 \\ \tr(Q \rho) \geq \epsilon}} \tr(Q \sigma) / \epsilon\ ,
\end{align}
where the optimisation is over operators $Q$. For $\rho$ a density operator,\footnote{We note that the definition proposed in~\cite{TCR09,Tom12} applies more generally to any $\rho$ with $\tr(\rho) \leq 1$.} the \emph{$\epsilon$-smooth max-relative entropy} is defined by
  \begin{align} \label{eq_Dmaxdef}
    2^{-D_{\max}^\epsilon(\rho\|\sigma)}
  =
    \sup_{\substack{\mu \bar{\rho} \leq \sigma \\ 1-F(\bar{\rho}, \rho)^2 \leq \epsilon^2}} \mu \ ,
  \end{align}
  where the optimisation is over non-negative operators $\bar{\rho}$ with $\tr(\bar{\rho}) \leq 1$.\footnote{The expressions on the right hand side of~\eqref{eq_DHdefinition} and~\eqref{eq_Dmaxdef} may be equal to~$0$, in which case the corresponding relative entropy is defined to be equal to $\infty$.}
\end{definition}

\begin{remark} \label{rem_DHrescaling}
  The second argument, $\sigma$, of the two one-shot entropy measures of Definition~\ref{def_oneshotentropy} may be rescaled easily because
  \begin{align} 
    D_H^\epsilon(\rho \| \lambda \sigma) & = D_H^\epsilon(\rho \| \sigma) - \log_2(\lambda)  \\
    D_{\max}^\epsilon(\rho \| \lambda \sigma) & = D_{\max}^\epsilon(\rho \| \sigma) - \log_2(\lambda)
  \end{align}  
  holds for any $\lambda > 0$. 
\end{remark}

The generalised relative entropy may be expressed equivalently as follows. 
\begin{lemma}
For any two non-negative operators $\rho$ and $\sigma$, 
\begin{align} \label{eq_DHdual}
2^{-D_H^\epsilon(\rho \| \sigma)} =  \sup_{\substack{\mu (\rho-Y) \leq \sigma  \\ Y \geq 0 \\ \mu \geq 0}} \mu(1- \tr(Y)/\epsilon) \ .
\end{align}
where the optimisation is over operators $Y$ and reals $\mu$. 
\end{lemma}
\begin{proof}
As shown in~\cite{DKFRR13}, the expression on the right hand side of~\eqref{eq_DHdefinition} is a semidefinite program whose dual has the form
\begin{align}
  2^{-D_H^\epsilon(\rho \| \sigma)} =  \sup_{\substack{\mu \rho \leq \sigma + X \\ X \geq 0 \\ \mu \geq 0}} \mu - \tr(X)/\epsilon \ ,
\end{align}
where the optimisation is over operators $X$ and reals $\mu$. Replacing $X$ by $\mu Y$ we can rewrite this as 
\begin{align} 
2^{-D_H^\epsilon(\rho \| \sigma)} = \sup_{\substack{\mu (\rho-Y) \leq \sigma  \\ \mu Y \geq 0 \\ \mu \geq 0}} \mu(1- \tr(Y)/\epsilon) \ .
\end{align}
To conclude the proof, we note that the condition $\mu Y \geq 0$ can be replaced by $Y \geq 0$ because $\mu \geq 0$ and because for $\mu = 0$ the value of $Y$ is irrelevant. 
\end{proof}

\begin{remark} \label{rem_DHepsilondependence}
It is obvious from this representation that $D_H^\epsilon(\rho \| \sigma)$ is a monotonically non-increasing function in $\epsilon$.
\end{remark}

The following lemma provides an upper bound on $D_H^\epsilon(\rho \| \sigma)$, expressed in terms of the trace distance of $\rho$ to an operator $\bar{\rho}$ that satisfies a simple operator inequality. 

\begin{lemma} \label{lem_DHupperbound}
  Let $\rho$, $\bar{\rho}$, and $\sigma$ be non-negative operators and suppose that $\bar{\rho} \leq \lambda \sigma$ for some $\lambda > 0$. Then
  \begin{align}
    D_H^\epsilon(\rho \| \sigma) \leq  \log_2(\lambda) - \log_2\bigl(1 - \Delta(\rho, \bar{\rho})/\epsilon\bigr)\ .
  \end{align}
\end{lemma} 

\begin{proof}
  Let $Y^+ \geq 0$ and $Y^- \geq 0$ be the positive and negative parts of $\rho - \bar{\rho}$, respectively, so that $ \rho - \bar{\rho} = Y^+ - Y^-$  and $\tr(Y^+)  \leq \Delta(\rho, \bar{\rho})$ (see Eq.~\ref{eq_tracedistancepositive}).  We then have, in particular, $\rho - Y^+ \leq \bar{\rho}$ and, using the assumption that $\bar{\rho} \leq \lambda \sigma$,
  \begin{align}
    \rho - Y^+ \leq  \lambda \sigma \ .
  \end{align}
  This means that $\mu = \frac{1}{\lambda}$ and $Y=Y^+$ fulfill the constraints of the maximisation in~\eqref{eq_DHdual} for $D_H^\epsilon(\rho \| \sigma)$ and, hence,
  \begin{align}
    2^{-D_H^\epsilon(\rho\|\sigma)} 
    \geq \frac{1}{\lambda} \bigl(1- \frac{\tr(Y^+)}{\epsilon} \bigr) 
    \geq \frac{1}{ \lambda}\bigl(1- \frac{\Delta(\rho, \bar{\rho})}{\epsilon}\bigr) \ .
  \end{align}
  Taking the negative logarithm on both sides of the inequality yields the claim.
\end{proof}

Although we are not using this for our argument, we note that Lemma~\ref{lem_DHupperbound} can be extended to a relation between the smooth relative max-entropy and the generalised relative entropy.

\begin{lemma} \label{lem_DHDmax}
  Let $\rho$ be a density operator, let $\sigma$ be a non-negative operator, and let $\epsilon > \epsilon' \geq 0$. Then
  \begin{align}
    D_H^{\epsilon}(\rho \| \sigma) \leq D_{\max}^{\epsilon'}(\rho \| \sigma) + \log_2 \frac{\epsilon}{\epsilon - \epsilon'} \ .
  \end{align}
\end{lemma}

\begin{proof}  
Let $\mu = 2^{-D_{\max}^{\epsilon'}(\rho \| \sigma)}$ and let $\bar{\rho}$ be such that the expression on the right hand side of~\eqref{eq_Dmaxdef} is satisfied for $D_{\max}^{\epsilon'}(\rho\|\sigma)$. That is, we have $\bar{\rho} \leq \sigma/\mu$ as well as $\sqrt{1-F(\bar{\rho}, \rho)^2} \leq \epsilon'$, which, by Lemma~\ref{lem_tracedistancefidelity},  implies  $\Delta(\bar{\rho}, \rho) \leq \epsilon'$. Hence, by Lemma~\ref{lem_DHupperbound},
\begin{align}
  D_H^\epsilon(\rho \| \sigma) 
\leq
   - \log_2(\mu) - \log_2 \bigl(1-(\epsilon'/\epsilon) \bigr)
=
  D_{\max}^{\epsilon'}(\rho\|\sigma) + \log_2 \frac{\epsilon}{\epsilon - \epsilon'} \ .  
\end{align}
\end{proof}

The following claim about the smooth relative max-entropy for product states is known as the \emph{Quantum Asymptotic Equipartition Property}~\cite{TCR09}. (Note that this is a strictly more general statement than the ``classical'' Asymptotic Equipartition Property stated as Lemma~\ref{lem_typicalsubspace}.) While the proof in~\cite{TCR09} applies to the case where $\sigma$ is a density operator, the slightly extended claim provided here  follows directly from Remark~\ref{rem_DHrescaling}  (see also Footnote~9 of~\cite{TCR09} as well as Chapter~6 of~\cite{Tom12}). 

\begin{lemma} \label{lem_DAEP}
  For any density operator $\rho$, for any non-negative operator $\sigma$, for any  $\epsilon \in (0,1)$, and for sufficiently large $n \in \mathbb{N}$, 
  \begin{align} \label{eq_DAEP}
    \frac{1}{n} D_{\max}^\epsilon(\rho^{\otimes n} \| \sigma^{\otimes n}) < D(\rho \| \sigma) + c \sqrt{\frac{ \log_2 (2 / \epsilon^2)}{n}} \ ,
  \end{align}
  where $c = c(\rho, \sigma)$ is independent of $n$ and $\epsilon$.
\end{lemma}
  
Because of Lemma~\ref{lem_DHDmax}, almost the same upper bound also holds for $D_H^\epsilon(\cdot \| \cdot)$. In fact, as a consequence of the Quantum Stein's Lemma~\cite{HP91,ON00}, the statement holds asymptotically with equality~\cite{DKFRR13}.
  
  \begin{lemma} \label{lem_QAEPequal}
  Let $\rho$ be a density operator, let $\sigma$ be a non-negative operator, and let $\epsilon \in (0,1)$. Then
\begin{align}
  \lim_{n \to \infty} \frac{1}{n} D_H^{\epsilon}(\rho^{\otimes n} \| \sigma^{\otimes n}) = D(\rho \| \sigma) \ .
\end{align}      
  \end{lemma}


\section{General facts about the fidelity} \label{app_fidelity}

In the literature, the definition and discussion of the fidelity $F(\rho, \sigma)$ is often restricted to the case where its arguments, $\rho$ and $\sigma$, are density operators  (see, e.g., Chapter~9 of~\cite{NC00}). In this work, however, we need the fidelity for general non-negative operators. Recall that we defined $F(\rho, \sigma) = \| \sqrt{\rho} \sqrt{\sigma} \|_1$. Fortunately, most established properties of the fidelity are still valid in this more general case. For completeness, we state them in the following. 

\begin{lemma} \label{lem_tracedistancefidelity}
For any two non-negative operators $\rho$ and $\sigma$ with $\tr(\rho) \geq \tr(\sigma)$, the trace distance is upper bounded by
  \begin{align} \label{eq_tracedistancefidelity}
  \Delta(\rho, \sigma) 
  \leq \sqrt{\tr(\rho)^2-F(\rho, \sigma)^2} \ . 
\end{align}
\end{lemma} 

\comment{
\begin{proof}
  Let $\omega$ be a non-negative operator with $\tr(\omega) = \tr(\rho) - \tr(\sigma)$, whose support is orthogonal to the support of both $\rho$ and $\sigma$, and define $\sigma' = \sigma + \omega$. Then $\tr(\rho) = \tr(\sigma')$ and
    \begin{align}
   \Delta(\rho, \sigma) = \Delta(\rho, \sigma')  \qquad \text{and} \qquad  F(\rho, \sigma) = F(\rho, \sigma')  \ .
  \end{align}
  It therefore suffices to show that the claim holds for operators with $\tr(\rho) = \tr(\sigma) = c$. Furthermore, defining $\bar{\rho} = \rho / c$ and $\bar{\sigma} = \sigma / c$ and noting that
  \begin{align}
    \Delta(\rho, \sigma) = c  \Delta(\bar{\rho}, \bar{\sigma}) \qquad \text{and} \qquad F(\rho, \sigma) = c F(\bar{\rho}, \bar{\sigma}) \ ,
  \end{align}  
  it suffices to verify that the claim holds for $\tr(\rho) = \tr(\sigma) = 1$. For a proof of this, see, e.g., \cite{NC00}. 
\end{proof}
}

The following lemma relates the relative entropy to the fidelity (see also Section~5.4 of~\cite{Hayashi06}). 
 
 \begin{lemma} \label{lem_DFidelity}
   For any non-negative operators $\rho$ and $\sigma$
   \begin{align}
      D(\rho \| \sigma)  \geq -2 \log_2 \frac{F(\rho , \sigma)}{\tr(\rho)}  \ .
   \end{align}
 \end{lemma}
 
 \begin{proof} 
   Let $D_\alpha(\cdot \| \cdot)$ be the \emph{$\alpha$-Quantum R\'enyi Divergence} as defined in~\cite{MDSFT13,WWY13}. As shown in these papers, for $\alpha = 1$ it is identical to the relative entropy, i.e., 
 \begin{align}
   D_1(\rho \| \sigma) = D(\rho \| \sigma) \ .
 \end{align}
 For $\alpha = 1/2$, it is related to the fidelity via
 \begin{align}
   D_{\frac{1}{2}}(\rho \| \sigma) = -2 \log_2 \frac{F(\rho, \sigma)}{\tr(\rho)} \ .
 \end{align}
Finally, $\alpha \mapsto D_{\alpha}(\rho \| \sigma)$ is a monotonically non-decreasing function in $\alpha$.  Combining these statements, we find 
 \begin{align}
 \label{eq_fidelityrelativeEntropy}
   -2 \log_2 F(\rho , \sigma) 
   = D_{\frac{1}{2}}(\rho \| \sigma) 
 \leq D_1(\rho \| \sigma)
 =  D(\rho \| \sigma) \ .
 \end{align}
 \end{proof}

Next we recall a statement that is known as Uhlmann's theorem~\cite{Uhl76}.

\begin{lemma} \label{lem_Uhlmann}
  Let $\rho_{D R} = \proj{\psi}_{D R}$ and $\sigma_{D R} = \proj{\phi}_{D R}$ be purifications of non-negative operators $\rho = \rho_D$ and $\sigma = \sigma_D$, respectively.  Then
  \begin{align}
    F(\rho, \sigma) = \sup_{U_R} \bigl| \bra{\psi} (\id_D \otimes U_R) \ket{\phi} \bigr| \ , 
  \end{align}
  where the maximisation is over all unitaries $U_R$ on $R$. 
\end{lemma}

\comment{
\begin{proof}
  Note that this claim can be  obtained directly from the corresponding standard statement where $\tr(\rho) = \tr(\sigma) = 1$. To see this, consider the normalised operators $\bar{\rho} = \rho / \tr(\rho)$ and $\bar{\sigma} = \sigma / \tr(\sigma)$. The corresponding purifications  $\proj{\bar{\psi}}$ and $\proj{\bar{\phi}}$ then satisfy $\ket{\psi} =\sqrt{\tr(\rho)} \ket{\bar{\psi}}$ and $\ket{\phi} = \sqrt{\tr(\sigma)} \ket{\bar{\phi}}$. Furthermore, $F(\rho, \sigma) = \sqrt{\tr(\rho) \tr(\sigma)} F(\bar{\rho}, \bar{\sigma})$.  
\end{proof}
}

The following lemma is a direct consequence of Lemma~\ref{lem_Uhlmann}. It asserts that the fidelity is monotonically non-decreasing when a partial trace is applied to both arguments. 

\begin{lemma} \label{lem_fidelitypartialtrace}
  For any two non-negative operators $\rho_{D E}$ and $\sigma_{D E}$ we have 
\begin{align}
  F(\rho_D, \sigma_D) \geq F(\rho_{D E}, \sigma_{D E}) \ .
\end{align}
\end{lemma}

\comment{
\begin{proof}
  Let $\proj{\psi}_{D E R}$ and $\proj{\phi}_{D E R}$ be purifications of $\rho_{D E}$ and $\sigma_{D E}$, respectively. Then, by Lemma~\ref{lem_Uhlmann},
  \begin{align}
    F(\rho_D, \sigma_D)
  = \sup_{U_{D R} } \bigl| \bra{\psi} (\id_D \otimes U_{D R}) \ket{\phi} \bigr| 
  \geq \sup_{U_R} \bigl| \bra{\psi} (\id_D \otimes \id_D \otimes U_{R}) \ket{\phi} \bigr| 
  = F(\rho_{D E}, \sigma_{D E}) \ .
  \end{align}
\end{proof}
}

Using the Stinespring dilation theorem, the statement can be brought into the following more general form.

\begin{lemma} \label{lem_fidelityTPCPM}
  For any trace-preserving completely positive map~$\mathcal{T}$ we have
  \begin{align}
  F(\mathcal{T}(\rho), \mathcal{T}(\sigma)) \geq F(\rho, \sigma) \ .
\end{align}
\end{lemma}

The next few claims allow us to keep track of the change of the  fidelity when we apply operators to its arguments. 

\begin{lemma}  \label{lem_fidelityoperator}
  For any non-negative operators $\rho$ and $\sigma$ and any operator $W$  on the same space we have
  \begin{align}
    F(\rho, W \sigma W^{\dagger}) = F(W^{\dagger} \rho W, \sigma) \ .
  \end{align}
\end{lemma}

\begin{proof}
  Let  $W_D = W$ and let $\proj{\psi}_{D R}$ and $\proj{\phi}_{D R}$ be purifications of $\rho_D = \rho$ and $\sigma_D = \sigma$, respectively. Then, by Uhlmann's theorem (Lemma~\ref{lem_Uhlmann}), 
  \begin{align}
    F(\rho_D, W_D \sigma_D W_D^{\dagger})
    =  \sup_{U_R} \bigl| \bra{\psi} (W_D \otimes U_R) \ket{\phi} \bigr|
    = \sup_{U_R}  \bigl| \bra{\phi} (W^{\dagger}_D \otimes U^{\dagger}_R) \ket{\psi} \bigr| 
    = F(\sigma_D, W^{\dagger}_D \rho_D W_D) \ ,
  \end{align}
  where the maximisation is taken over the set of unitaries $U_R$ on $R$.   
  \end{proof}

\begin{lemma} \label{lem_fidelitydecomposition}
  Let $\rho$ and $\sigma$ be non-negative operators and let $\{W_d\}_{d \in D}$ be a family of operators such that $\sum_{d \in D} W_d = \id$. Then
  \begin{align}
    \sum_{d \in D} F(W_d^{\dagger} \rho W_d, \sigma) \geq F(\rho, \sigma) \ .
  \end{align}
\end{lemma}

\begin{proof}
Let $\proj{\psi}_{D R}$ and $\proj{\phi}_{D R}$ be purifications of $\rho_D = \rho$ and $\sigma_D = \sigma$, respectively.  By Uhlmann's theorem (Lemma~\ref{lem_Uhlmann}), there exists a unitary $U_R$ on $R$ such that
  \begin{align}
    F(\rho, \sigma) 
    = \bigl|\bra{\psi} (\id_D \otimes U_R) \ket{\phi} \bigr|
    = \Bigl| \sum_{d \in D}  \bra{\psi} (W_d \otimes U_R) \ket{\phi} \Bigr|
    \leq \sum_{d \in D} \bigl| \bra{\psi} (W_d \otimes U_R) \ket{\phi} \bigr| \ .
  \end{align}
  The assertion follows because, again by Uhlmann's theorem,
  \begin{align}
    \bigl| \bra{\psi} (W_d \otimes U_R) \ket{\phi} \bigr| \leq F(W_d^{\dagger} \rho W_d, \sigma) 
  \end{align}  
  holds for any $d \in D$. 
\end{proof}

\begin{lemma} \label{lem_fidelityreduced}
  Let $\rho_{D E}$ and $\sigma_{D E}$ be non-negative operators on $D \otimes E$ and let $\mathcal{W}_E$ be a trace non-increasing completely positive map on $E$. Then
  \begin{align} \label{eq_fidelityreduced}
    F\bigl(\rho_{D E}, (\cI_D \otimes \mathcal{W}_E)(\sigma_{D E})  \bigr)
    \leq F(\rho_D, \sigma_D) \ .
  \end{align}
\end{lemma}

\begin{proof}
  Let $X_E \mapsto \sum_e W_e X W_e^{\dagger}$ be an operator-sum representation of $\mathcal{W}_E$.   The second argument of the fidelity on the left hand side of~\eqref{eq_fidelityreduced}, $\sigma'_{D E}  = {(\cI_D \otimes \mathcal{W}_E)}(\sigma_{D E}) $, may then be written as
  \begin{align}
    \sigma'_{D E}
    = \sum_e (\id_D \otimes W_{e}) \sigma_{D E} (\id_D \otimes W_{e}^{\dagger})  \ .
  \end{align}
  Because, by assumption, $\sum_e W_{e}^{\dagger} W_{e} \leq \id_E$, we have 
    \begin{align}
    \sigma'_D
    = \tr_E(\sigma'_{D E})
    = \tr_E\bigl( \sum_e (\id_D \otimes W_{e}^\dagger W_{e}) \sigma_{D E} \bigr)
    \leq \tr_E(\sigma_{D E}) = \sigma_D \ .
  \end{align}
  Together with the fact that the square root is operator monotone (cf.\ Theorem~V.1.9 of~\cite{Bhatia97}), this implies
  \begin{align}
    F(\rho_D, \sigma'_D) 
  = \tr\bigl(\sqrt{\sqrt{\rho_D} \sigma'_D \sqrt{\rho_D}} \bigr)
  \leq  \tr\bigl(\sqrt{\sqrt{\rho_D} \sigma_D \sqrt{\rho_D}} \bigr)
  = F(\rho_D, \sigma_D) \ .
  \end{align}
  The claim then follows from Lemma~\ref{lem_fidelitypartialtrace}, which asserts that $F(\rho_{DE}, \sigma'_{D E}) \leq F(\rho_D, \sigma'_D)$. 
\end{proof}

We also recall that the fidelity is continuous in its arguments with respect to the trace norm.

\begin{lemma} \label{lem_fidelitycontinuoustracenorm}
  Let $\rho$, $\rho'$, and $\sigma$ be non-negative operators. Then
  \begin{align}
    \bigl| F(\rho, \sigma) - F(\rho', \sigma) \bigr|^2 \leq \| \rho - \rho' \|_1 \tr(\sigma) \ .
  \end{align}
\end{lemma}

\begin{proof}
  Let $\ket{\phi}_{D R}$ be a purification of $\sigma_D = \sigma$. Furthermore, let $\ket{\psi}_{D R}$, $\ket{\psi'}_{D R}$ be purifications of $\rho_D = \rho$, $\rho'_D = \rho'$ such that $F(\rho, \rho') = \bigl| \spr{\psi}{\psi'} \bigr|$ (cf.\ Lemma~\ref{lem_Uhlmann}) and assume without loss of generality that $\spr{\psi}{\psi'} \geq 0$. We have
  \begin{align}
    F(\rho, \sigma) - F(\rho', \sigma)
    = \sup_{U} \bigl|\bra{\psi} (\id_D \otimes U_R) \ket{\phi}\bigr| - \sup_{U'} \bigl|\bra{\psi'} (\id_D \otimes U'_R) \ket{\phi}\bigr| \\
    \leq \sup_{U} \bigl|\bra{\psi} (\id_D \otimes U_R) \ket{\phi}\bigr| - \bigl|\bra{\psi'} (\id_D \otimes U_R) \ket{\phi}\bigr| \\
    \leq \sup_{U} \bigl|(\bra{\psi} - \bra{\psi'}) (\id_D \otimes U_R) \ket{\phi}\bigr| \\
    \leq \sup_{U} \bigl\| \ket{\psi} - \ket{\psi'} \bigr\|_2 \, \bigl\| (\id_D \otimes U_R) \ket{\phi} \bigr\|_2  \\
    = \bigl\| \ket{\psi} - \ket{\psi'} \bigr\|_2 \, \bigl\| \ket{\phi} \bigr\|_2 \ ,
  \end{align}
  where we have used the Cauchy-Schwarz inequality.  The claim then follows from $\| \ket{\phi} \|_2^2 = \tr(\sigma)$ and
  \begin{align} \label{eq_normsinequality}
    \bigl\| \ket{\psi} - \ket{\psi'} \bigr\|_2^2 = \tr(\rho) + \tr(\rho') - 2 \spr{\psi}{\psi'} = \tr(\rho) + \tr(\rho') - 2 F(\rho, \rho') \leq \|\rho - \rho' \|_1 \ ,
       \end{align}
   where we have used the Fuchs-van de Graaf inequality~\cite{FvdG99}.
     \end{proof}

Finally, we provide a lemma (Lemma~\ref{lem_fidelitysymmetricpurification}) that simplifies the evaluation of the fidelity between permutation-invariant operators. It may be seen as a generalisation of a known result on symmetric purifications, which we state as Lemma~\ref{lem_symmetricpurification} (see, e.g., Lemma~II.5 of~\cite{CKMR07} for a proof). Specifically, Lemma~\ref{lem_fidelitysymmetricpurification} may be seen as a combination of this result and Uhlmann's theorem (Lemma~\ref{lem_Uhlmann}).  We also note that the lemma may be generalised to other symmetry groups (other than the symmetric group). 

\begin{lemma} \label{lem_symmetricpurification}
For any permutation-invariant operator $\rho_{D^n}$ on $D^{\otimes n}$ and any space $R$ with  $\dim(R) \geq \dim(D)$ there exists a  permutation-invariant purification $\rho_{D^n R^n}$ on $(D \otimes R)^{\otimes n}$.
\end{lemma}

\begin{lemma} \label{lem_fidelitysymmetricpurification}
  Let  $\rho_{D^n}$ and $\sigma_{D^n}$  be  permutation-invariant non-negative operators on  $D^{\otimes n}$ and let $\rho_{D^n R^n}$ be a permutation-invariant purification of $\rho_{D^n}$. Then there exists a permutation-invariant purification  $\sigma_{D^n R^n}$ of $\sigma_{D^n}$ such that
  \begin{align} \label{eq_fidelitysymmetric}
    F(\rho_{D^n}, \sigma_{D^n}) = F(\rho_{D^n R^n}, \sigma_{D^n R^n}) \ .
  \end{align} 
\end{lemma}

\begin{proof}
  The proof of this lemma essentially follows the lines of the standard proof of Uhlmann's theorem (see, e.g., Chapter~9 of~\cite{NC00}), while keeping track of the permutation invariance the relevant operators.
  
For the following, we assume without loss of generality that $\rho_{D^n}$ and $\sigma_{D^n}$ are invertible. (The claim for the cases where this assumption does not hold may be obtained by considering the operators $\rho_{D^n} + \epsilon \, \id_{D^n}$ and $\sigma_{D^n} + \epsilon \,  \id_{D^n}$ for $\epsilon > 0$ and then taking the limit $\epsilon \to 0$.)

Let $\ket{\Psi}_{D^n R^n}$ be a vector in $(D \otimes R)^{\otimes n}$ such that $\rho_{D^n R^n} = \proj{\Psi}_{D R}$ and define
\begin{align} 
  \ket{\Omega}_{D^n R^n} = (\rho_{D^n}^{-\frac{1}{2}} \otimes \id_{R^n}) \ket{\Psi}_{D^n R^n} \ .
\end{align}
Note that $\tr_{R^n}(\proj{\Omega}_{D^n R^n}) = \id_{D^n}$. It thus follows from the Schmidt decomposition that $\ket{\Omega}_{D^n R^n}$ has the form
\begin{align} \label{eq_Omegadecomposition}
  \ket{\Omega}_{D^n R^n} = \sum_{x} \ket{d_x}_{D^n} \otimes \ket{r_x}_{R^n} \ ,
\end{align}
where $\{\ket{d_x}_{D^n}\}_x$ and $\{\ket{r_x}_{R^n}\}_x$ are orthonormal bases of $D^{\otimes n}$ and $R^{\otimes n}$, respectively. 

Let $U_{D^n}$ be the unitary operator in the left polar decomposition of $\sqrt{\rho_{D^n}} \sqrt{\sigma_{D^n}}$, i.e., 
\begin{align} \label{eq_Udecompositiondef}
  \sqrt{\rho_{D^n}} \sqrt{\sigma_{D^n}} = Q_{D^n}  U_{D^n}\ ,
\end{align}
where 
\begin{align}
  Q_{D^n} =  \sqrt{( \sqrt{\rho_{D^n}} \sqrt{\sigma_{D^n}} ) ( \sqrt{\rho_{D^n}} \sqrt{\sigma_{D^n}} )^{\dagger}} = \sqrt{\sqrt{\rho_{D^n}} \sigma_{D^n} \sqrt{\rho_{D^n}}} 
  \end{align}
  is non-negative. We now define the purification $\sigma_{D^n R^n} = \proj{\Phi}_{D^n R^n}$ by 
\begin{align}
  \ket{\Phi}  = (\sqrt{\sigma_{D^n}} U^{\dagger}_{D^n} \otimes \id_{R^n}) \ket{\Omega}  \ .
\end{align}
It is readily verified that this is indeed a purification of $\sigma_{D^n}$. 

The fidelity between the purifications is given by
\begin{align}
  F(\rho_{D^n R^n}, \sigma_{D^n R^n})
  = \bigl| \spr{\Psi}{\Phi} \bigr|
  = \bigl| \bra{\Omega} (\sqrt{\rho_{D^n}} \sqrt{\sigma_{D^n}} U^{\dagger} \otimes \id_{R^n})  \ket{\Omega} \bigr| \ .
\end{align}
Exploiting now the particular form~\eqref{eq_Omegadecomposition} of $\ket{\Omega}$ as well as~\eqref{eq_Udecompositiondef}, this can be rewritten as
\begin{align}
  F(\rho_{D^n R^n}, \sigma_{D^n R^n})
  = \bigl| \sum_x \bra{d_x} \sqrt{\rho_{D^n}} \sqrt{\sigma_{D^n}} U^{\dagger}_{D^n} \ket{d_x} \bigr|
  = \bigl| \tr(\sqrt{\rho_{D^n}} \sqrt{\sigma_{D^n}} U^{\dagger}_{D^n}) \bigr|
  = \tr(Q_{D^n})  \ .
\end{align}
The claim~\eqref{eq_fidelitysymmetric} then follows by inserting the explicit expression for $Q_{D^n}$, i.e., 
\begin{align}
  \tr(Q_{D^n})
=  \bigl\| \sqrt{\rho_{D^n}} \sqrt{\sigma_{D^n}} \bigr\|_1
=  F(\rho_{D^n}, \sigma_{D^n})  \ .
\end{align}

To verify that $\sigma_{D^n R^n}$ is permutation-invariant, we first note that for any permutation-invariant Hermitian operator $X$ on an $n$-fold product space and for any real function $f$ the operator $f(X)$ is also permutation-invariant. (To see this, consider the decomposition $X = \sum_i x_i \Pi_i$, where $\Pi_i$ are the projectors onto the eigenspaces of $X$ and $x_i$ are the corresponding eigenvalues.  Because $[X, \pi] = 0$ for any permutation $\pi$, we also have $[\Pi_i, \pi] = 0$ for any $i$. Using now that $f(X) = \sum_{i} f(x_i) \Pi_i$, we conclude that $[f(X), \pi] = 0$.)  We therefore know, in particular, that
\begin{align}
  [\rho_{D^n}^{\frac{1}{2}}, \pi] = 0 \qquad \text{and}  \qquad [\rho_{D^n}^{-\frac{1}{2}}, \pi] = 0 \qquad \text{and} \qquad   [\sigma_{D^n}^{\frac{1}{2}}, \pi] = 0  \qquad \text{and} \qquad [\sigma_{D^n}^{-\frac{1}{2}}, \pi] = 0  
\end{align}
for any permutation $\pi$. Furthermore, it follows from the explicit expression for $Q_{D^n}$ that this operator is also permutation-invariant. Similarly,  since $U^{\dagger}_{D^n}$ can be written as
\begin{align}
  U^{\dagger}_{D^n} = \sigma_{D^n}^{-\frac{1}{2}} \rho_{D^n}^{-\frac{1}{2}} Q_{D^n}  \ ,
\end{align}
it is also permutation-invariant. By assumption, we also have $\pi \ket{\Psi}_{D^n R^n} = \ket{\Psi}_{D^n R^n}$. Because $\ket{\Phi}_{D^n R^n}$ is obtained by multiplying permutation-invariant operators to $\ket{\Psi}_{D^n R^n}$, we conclude that $\pi \ket{\Phi}_{D^n R^n} = \ket{\Phi}_{D^n R^n}$, i.e., $\sigma_{D^n R^n}$ is invariant under permutations. 
\end{proof}

\section{On the Schur-Weyl duality} \label{app_representation} 

The following lemma follows immediately from the considerations in Chapter~6 of~\cite{Harrow05} (see, in particular, Eq.~6.25). 

\begin{lemma} \label{lem_SchurWeyldecomposition}
  Let $D$ and $E$ be Hilbert spaces with $\dim(D) = \dim(E) = d$ and let $n \in \mathbb{N}$. Furthermore, let $\Lambda_{n,d}$ be the set of Young diagrams of size $n$ with at most $d$ rows, and, for any $\lambda \in \Lambda_{n,d}$, let $U_\lambda$ and $V_\lambda$ be the corresponding irreducible representations of the unitary group $U(d)$ and the symmetric group $S_n$, respectively, so that, according to the Schur-Weyl duality (see, e.g., Theorem~1.10 of~\cite{Christandl06})
  \begin{align} \label{eq_decompositionDn}
   D^{\otimes n}  & \cong \bigoplus_{\lambda \in \Lambda_{n,d}} U_{D, \lambda} \otimes V_{D, \lambda} \\  \label{eq_decompositionEn}
   E^{\otimes n}  & \cong \bigoplus_{\lambda \in \Lambda_{n,d}} U_{E, \lambda} \otimes V_{E, \lambda} \ .
  \end{align}
  Then there exists a family $\{\ket{\psi_\lambda}_{V_{D, \lambda} V_{E, \lambda}}\}_{\lambda \in \Lambda_{n,d}}$ of maximally entangled normalised vectors   on $V_{D, \lambda} \otimes V_{E, \lambda}$ such that any vector $\ket{\Omega} \in \Sym^n(D \otimes E)$ in the symmetric subspace of $(D \otimes E)^{\otimes n}$ can be decomposed as
  \begin{align} \label{eq_symmetricdecomposition}
    \ket{\Omega} = \sum_{\lambda} \ket{\phi_\lambda}_{U_{D, \lambda} U_{E, \lambda}} \otimes \ket{\psi_\lambda}_{V_{D, \lambda} V_{E, \lambda}} \ ,
  \end{align}
  where $\{\ket{\phi_\lambda}_{U_{D, \lambda}  U_{E, \lambda}}\}_{\lambda \in \Lambda_{n,d}}$ is a family of (not necessarily normalised) vectors on $U_{D, \lambda} \otimes U_{E, \lambda}$.
\end{lemma}

\begin{proof}
  According to~\eqref{eq_decompositionDn} and~\eqref{eq_decompositionEn}, the space $(D \otimes E)^{\otimes n}$ decomposes as
  \begin{align}
    (D \otimes E)^{\otimes n} \cong \Bigl( \bigoplus_{\lambda \in \Lambda_{n,d}} U_{D, \lambda} \otimes V_{D, \lambda}  \Bigr) \otimes \Bigl( \bigoplus_{\lambda' \in \Lambda_{n,d}} U_{E, \lambda'} \otimes V_{E, \lambda'}  \Bigr) \ .
  \end{align}
  Any vector $\ket{\Omega} \in (D \otimes E)^{\otimes n}$ can therefore always be written as
    \begin{align} \label{eq_Omegageneral}
    \ket{\Omega} = \sum_{\lambda, \lambda' \in \Lambda_{n,d}} \sum_i  \ket{\phi_{\lambda, \lambda', i}}_{U_{D, \lambda} U_{E, \lambda'}} \otimes \ket{\psi_{\lambda, \lambda', i}}_{V_{D, \lambda} V_{E, \lambda'}} \ ,
  \end{align}
  where, for any $\lambda, \lambda' \in \Lambda_{n,d}$, $\{\ket{\phi_{\lambda, \lambda', i}}_{U_{D, \lambda} U_{E, \lambda'}}\}_i$ and $\{\ket{\psi_{\lambda, \lambda', i}}_{V_{D, \lambda} V_{E, \lambda'}}\}_i$ are families of vectors in ${U_{D, \lambda} \otimes U_{E, \lambda'}}$ and ${V_{D, \lambda} \otimes V_{E, \lambda'}}$, respectively. 
  
For any $\lambda \in \Lambda_{n, d}$, let  $\{\ket{v_k}_{V_{D, \lambda}}\}_k$ and $\{\ket{\bar{v}_k}_{V_{E, \lambda}}\}_k$ be orthonormal bases of $V_{D, \lambda}$ and $V_{E, \lambda}$, respectively, with respect to which the representations of the symmetric group $S_n$ are given by the same real-valued matrices. (Such bases always exist, see, e.g., \cite{JamKer81}.) We then define the maximally entangled vector $\ket{\psi_\lambda}_{V_{D, \lambda} V_{E, \lambda}}$ on $V_{D, \lambda} \otimes V_{E, \lambda}$ by
\begin{align}
  \ket{\psi_\lambda}_{V_{D, \lambda} V_{E, \lambda}}
  =  \sqrt{\frac{1}{\dim {V_\lambda}}} \sum_{k} \ket{v_k}_{V_{D, \lambda}} \otimes \ket{\bar{v}_k}_{V_{E, \lambda}} \ .
\end{align}
  Now, to prove the claim~\eqref{eq_symmetricdecomposition} for any permutation-invariant $\ket{\Omega}$, it suffices to show that the  vectors on $V_{D, \lambda} \otimes V_{E, \lambda'}$ in~\eqref{eq_Omegageneral} satisfy 
  \begin{align} \label{eq_psilambdastructure}
    \ket{\psi_{\lambda, \lambda', i}}_{V_{D, \lambda} V_{E\lambda'}} = \begin{cases} \ket{\psi_\lambda}_{V_{D, \lambda} V_{E, \lambda}} & \text{if $\lambda = \lambda'$} \\ 0 & \text{otherwise} \end{cases} 
  \end{align}
   for all $\lambda, \lambda' \in \Lambda_{n,d}$ and for all $i$. 
  
For any permutation $\pi$, let $V_{\lambda}(\pi)$ be its action on the irreducible space $V_{\lambda}$.  Using that, by definition, the matrix elements $\bra{v_k} V_{D, \lambda}(\pi) \ket{v_{k'}} = \bra{\bar{v}_k} V_{E, \lambda}(\pi) \ket{\bar{v}_{k'}}$ are real-valued, it is easily verified that
    \begin{align} \label{eq_symmetryentangled}
     (V_{D, \lambda}(\pi) \otimes V_{E, \lambda}(\pi)) \ket{\psi_{\lambda}}_{V_{D, \lambda} V_{E, \lambda}} = \ket{\psi_{\lambda}}_{V_{D, \lambda} V_{E, \lambda}}  \quad (\forall \pi) \ .
  \end{align}
  Furthermore, the Schur-Weyl duality (cf.\ Theorem~1.10 of~\cite{Christandl06}) states that $\pi$ acts on $\bigoplus_{\lambda} U_{\lambda} \otimes V_{\lambda}$ as
  \begin{align}
    V(\pi) = \bigoplus_{\lambda \in \Lambda_{n,d}} \id_{U_{\lambda}} \otimes V_{\lambda}(\pi) \ ,
  \end{align}
  Using this and that the vector $\ket{\Omega}$ is by assumption invariant under the action of  $\pi$, we find that
  \begin{align} \label{eq_symmetryonlambda}
     (V_{D, \lambda}(\pi) \otimes V_{E, \lambda'}(\pi)) \ket{\psi_{\lambda, \lambda', i}}_{V_{D, \lambda} V_{E, \lambda'}} = \ket{\psi_{\lambda, \lambda', i}}_{V_{D, \lambda} V_{E, \lambda'}}  \quad (\forall \pi) 
  \end{align}
  holds for all $\lambda, \lambda' \in \Lambda_{n,d}$ and for all $i$ for which the corresponding term in the sum~\eqref{eq_Omegageneral} is nonzero.  
  
  For any such triple $(\lambda, \lambda', i)$ let $H_{\lambda, \lambda', i}$ be the homomorphism between the irreducible representations $V_{E, \lambda'}$ and $V_{E, \lambda}$ defined by
  \begin{align}  \label{eq_Hdef}
    \bra{\alpha} H_{\lambda, \lambda', i}  \ket{\beta}
  =  \bra{\psi_{\lambda, \lambda', i}} (\id_{D, \lambda} \otimes \ket{\beta}\bra{\alpha}) \ket{\psi_{\lambda}}   \ ,
  \end{align}
  for any $\ket{\alpha} \in V_{E, \lambda}$, $\ket{\beta} \in V_{E, \lambda'}$. Using~\eqref{eq_symmetryonlambda} and~\eqref{eq_symmetryentangled} we find that, for any permutation $\pi$, 
\begin{multline}
  \bra{\alpha} H_{\lambda, \lambda', i}  V_{E, \lambda'}(\pi) \ket{\beta}
    =  \bra{\psi_{\lambda, \lambda', i}} \bigl(\id_{D, \lambda} \otimes V_{E, \lambda'}(\pi) \ket{\beta}\bra{\alpha}\bigr) \ket{\psi_{\lambda}} \\
    =  \bra{\psi_{\lambda, \lambda', i}} \bigl(V_{D, \lambda}(\pi)^{\dagger} \otimes \ket{\beta}\bra{\alpha}\bigr) \ket{\psi_{\lambda}} \\
    = \bra{\psi_{\lambda, \lambda', i}} \bigl(\id_{D, \lambda} \otimes \ket{\beta}\bra{\alpha} V_{E, \lambda}(\pi)\bigr) \ket{\psi_{\lambda}} 
    =  \bra{\alpha} V_{E, \lambda}(\pi) H_{\lambda, \lambda', i}  \ket{\beta} \ .
\end{multline}
  This implies that $H_{\lambda, \lambda', i} V_{E, \lambda'}(\pi) = V_{E, \lambda}(\pi) H_{\lambda, \lambda', i}$, i.e., $H_{\lambda, \lambda', i}$ commutes with the action of the symmetry group. Hence, by Schur's lemma (see, e.g., Lemma~0.8 of~\cite{Christandl06}) and the fact that the representations $V_{E, \lambda}$ and $V_{E, \lambda'}$ are inequivalent for $\lambda \neq \lambda'$, we find
  \begin{align}
    H_{\lambda, \lambda', i} = c_{\lambda, i} \delta_{\lambda, \lambda'} \id_{V_{E, \lambda}} \ ,
  \end{align}
  for some appropriately chosen coefficients $c_{\lambda, i}$.  Using~\eqref{eq_Hdef} with $\ket{\alpha} = \ket{\bar{v}_k}_{V_{E, \lambda}}$ and $\ket{\beta} = \ket{\bar{v}_{k'}}_{V_{E, \lambda'}}$ we obtain
  \begin{align}
    c_{\lambda, i} \delta_{\lambda, \lambda'} \delta_{k, k'} 
    =  \bra{\bar{v}_k} H_{\lambda, \lambda', i} \ket{\bar{v}_{k'}} 
     =       \bra{\psi_{\lambda, \lambda', i}} (\id_{D, \lambda} \otimes \ket{\bar{v}_{k'}}\bra{\bar{v}_k}) \ket{\psi_{\lambda}}      
     = \frac{1}{\sqrt{\dim(V_\lambda)}} \bra{\psi_{\lambda, \lambda', i}} (\ket{v_k} \otimes \ket{\bar{v}_{k'}})  \ .
  \end{align}
  Since this holds for any $k, k'$, we conclude that $\ket{\psi_{\lambda, \lambda', i}}_{V_{D, \lambda} V_{E, \lambda'}}$ is proportional to $\ket{\psi_\lambda}_{V_{D, \lambda} V_{E, \lambda}}$ if $\lambda = \lambda'$ and $0$ otherwise. Note that, for $\lambda = \lambda'$, the corresponding proportionality constant can without loss of generality be absorbed in $\ket{\phi_{\lambda, \lambda', i}}_{U_{D, \lambda} U_{E, \lambda'}}$ in the sum~\eqref{eq_Omegageneral}, so that $\ket{\psi_{\lambda, \lambda, i}}_{V_{D, \lambda} V_{E, \lambda}}$ is normalised. Hence, noting that $\ket{\psi_\lambda}_{V_{D, \lambda} V_{E, \lambda}}$ is also normalised, we have established~\eqref{eq_psilambdastructure}.
    \end{proof}
        
 \section{Strong faithfulness of squashed entanglement} \label{app_squashed}
  
As an example for how our result can be applied, we present here an argument proposed by Li and Winter~\cite{WL}. The argument is described in detail in~\cite{LW14}. We summarise it here for convenience.

\emph{Squashed entanglement} is a measure of entanglement defined for any bipartite state $\rho_{A C}$ as
 \begin{align}
   E_{\mathrm{sq}}(\rho_{A C}) = \frac{1}{2} \inf_{\rho_{A C E}} I(A : C | E)_{\rho} \ ,
 \end{align}
 where the infimum ranges over all non-negative extensions $\rho_{A C E}$ of $\rho_{A C}$~\cite{ChrWin03}. It is known that squashed entanglement is \emph{faithful}, i.e., strictly positive for any entangled state~\cite{BCY11,LiWin14}. In other words, $E_{\mathrm{sq}}(\rho_{A C})  = 0$ if and only if the state $\rho_{A C}$ is separable.  Theorem~\ref{thm_maininequality} implies a novel quantitative version of this claim. The main idea is to relate $E_{\mathrm{sq}}(\rho_{A C})$ to the distance between $\rho_{A C}$ and the closest state that is $k$-extendible (see Footnote~\ref{ftn_extendible} for a definition).
  
 \begin{theorem}[\cite{LW14}] \label{thm_squashed}
   For any density operator $\rho_{A C}$ on $A \otimes C$  and any $k \in \mathbb{N}$ there exists a $k$-extendible density operator $\omega_{A C}$ such that\footnote{We formulate the claim here for the trace distance~$\Delta(\cdot, \cdot)$, but note that it also holds for the purified distance defined in~\cite{TCR09}.}
   \begin{align} \label{eq_squashedextendiblebound}
     \Delta(\rho_{A C},  \omega_{A C}) \leq (k-1) \sqrt{\frac{\ln 2}{2} E_{\mathrm{sq}}(\rho_{A C})}  \ .
   \end{align}
 \end{theorem}

\begin{proof} 
 Let $\rho_{A C E}$ be a non-negative extension of $\rho_{A C}$. Theorem~\ref{thm_maininequality} implies that there exists a trace-preserving completely positive reconstruction map  $\cT_{E \to C E}$ such that
 \begin{align} \label{eq_reconstrdist}
   \Delta\bigl(\rho_{A C E}, (\cI_A \otimes \cT_{E \to C E})(\rho_{A E})\bigr) \leq \delta = \sqrt{\ln(2) I(A : C | E)_\rho} 
 \end{align}
 (see Eq.~\ref{eq_maininequalityc}).
 
 For $i \in \mathbb{N}$, define $\rho^i_{A C_1 \cdots C_i E}$ inductively by
 \begin{align} \label{eq_rhoidef}
   \rho^{i+1}_{A C_1 \cdots C_{i+1} E} = (\cI_{A C_1 \cdots C_i} \otimes \cT_{E \to C_{i+1} E})(\rho^i_{A C_1 \cdots C_i E}) \ ,
 \end{align}
  and $\rho^1_{A C_1 E} = \rho_{A C E}$.  Because the trace distance cannot increase under the action of $\cT_{E \to C E}$ (see Eq.~\ref{eq_tracedistancemonotone}) we  have
 \begin{align}
  \Delta(\rho^{i}_{A C_i E}, \rho^{i+1}_{A C_{i+1} E})
  \leq \Delta(\rho^{i-1}_{A E}, \rho^{i}_{A E})
  \leq \Delta(\rho^{i-1}_{A C_{i-1} E}, \rho^{i}_{A C_i E}) 
 \end{align}
 for $i > 1$. Furthermore, from~\eqref{eq_reconstrdist} we have
 \begin{align}
     \Delta(\rho^{1}_{A C_1 E}, \rho^{2}_{A C_{2} E})
     = \Delta\bigl(\rho_{A C E}, (\cI_{A} \otimes \cT_{E \to C E})(\rho_{A E}) \bigr)     
      \leq \delta \ .
 \end{align}
 The combination of these inequalities yields
 \begin{align}
    \Delta(\rho^{i}_{A C_i E}, \rho^{i+1}_{A C_{i+1} E}) \leq \delta
 \end{align}
 for any $i \in \mathbb{N}$. We now apply the triangle inequality to conclude that
 \begin{align} 
   \Delta( \rho_{A C E}, \rho^{j}_{A C_{j} E})
   =  \Delta( \rho^1_{A C_1 E}, \rho^{j}_{A C_{j} E})
   \leq \sum_{i=1}^{j-1} \Delta(\rho^{i}_{A C_{i} E}, \rho^{i+1}_{A C_{i+1} E}) = (j-1) \delta
 \end{align}
 for any $j \in \mathbb{N}$. Furthermore, because in the definition of the density operators $\rho^i_{A C_1 \cdots C_i E}$ (see Eq.~\ref{eq_rhoidef}) the reconstruction map does not act on the systems $C_1, \ldots C_i$, we have $\rho^j_{A C_1 \cdots C_j} = \rho^k_{A C_1 \cdots C_j}$ for any $j \leq k$, and hence
 \begin{align} \label{eq_rhojbound}
    \Delta( \rho_{A C}, \rho^{k}_{A C_{j}})
    = \Delta( \rho_{A C}, \rho^{j}_{A C_{j}})
  \leq (j-1) \cdot \delta \ .
 \end{align}
 
 Define now the density operator
 \begin{align}
   \bar{\omega}_{A C_1 \cdots C_k} = \frac{1}{k!} \sum_{\pi} \rho^k_{A C_{\pi(1)} \cdots C_{\pi(k)}} \ ,
 \end{align}
 where the sum ranges over all permutations of $\{1, \ldots, k\}$. The density operator $\omega_{A C} = \bar{\omega}_{A C_1}$ is then $k$-extendible by construction.  Using the convexity of the trace distance we find
  \begin{align}
   \Delta(\rho_{A C}, \omega_{A C})  
   =  \Delta(\rho_{A C}, \bar{\omega}_{A C_1})  
   \leq  \frac{1}{k!} \sum_{\pi} \Delta(\rho_{A C}, \rho^k_{A C_{\pi(1)}})
   =  \frac{1}{k} \sum_{j=1}^k \Delta(\rho_{A C}, \rho^{k}_{A C_j}) \ .
   \end{align}
   Inserting now the bound~\eqref{eq_rhojbound} we conclude that
 \begin{align}
  \Delta(\rho_{A C}, \omega_{A C})  
   \leq \frac{1}{k} \sum_{j=1}^k (j-1) \cdot \delta
   \leq \frac{k-1}{2} \cdot \delta 
   = \frac{k-1}{2}  \sqrt{\ln(2) I(A : C | E)_\rho}   \ .
 \end{align}
 The claim of the theorem follows because the above holds for any non-negative extension $\rho_{A C E}$ of $\rho_{A C}$.  
  \end{proof} 
 
   We remark that the bound provided by Theorem~\ref{thm_squashed} does not depend on the dimension of the two subsystems $A$ and $C$.  As mentioned above, this yields a quantitative claim on the faithfulness of squashed entanglement, which we formulate  as Corollary~\ref{cor_squashed} below. Its proof uses the following statement about the distance of $k$-extendible states from the set of separable states, which we denote by $S_{A : C}$. 
   
  \begin{lemma} \label{lem_extendibleseparable}
   For any $k$-extendible density operator $\omega_{A C}$ on $A \otimes C$
   \begin{align}
     \inf_{\sigma_{A C} \in S_{A : C}}  \Delta(\omega_{A C}, \sigma_{A C}) 
   \leq 2 \frac{(\dim C)^2}{k} \ .
   \end{align}
 \end{lemma}
 \begin{proof}
   By definition, there exists a density operator $\bar{\omega}_{A C_1 \cdots C_k}$ such that $\omega_{A C} = \bar{\omega}_{A C_i}$ for  $i = 1, \ldots, k$. Because this condition still holds if the order of the  subsystems $C_1, \ldots, C_n$ is permuted, one can assume without loss of generality that $\bar{\omega}_{A C_1 \ldots C_k}$ is invariant under such permutations. The claim then follows immediately from Theorem~II.$7'$ of~\cite{CKMR07}. 
 \end{proof}
      
 \begin{corollary}[\cite{LW14}] \label{cor_squashed}
   For any density operator $\rho_{A C}$  on $A \otimes C$
   \begin{align} \label{eq_squashedfaithful}
     \inf_{\sigma_{A C} \in S_{A : C}} \Delta(\rho_{A C}, \sigma_{A C}) \leq 2 \dim C \sqrt[4]{2 \ln(2) E_{\mathrm{sq}}(\rho_{A C})} \ .
   \end{align}
 \end{corollary}
  
  \begin{proof}
   Let $\omega_{A C}$ be a $k$-extendible density operator that satisfies~\eqref{eq_squashedextendiblebound}. Using the triangle inequality we can combine this with Lemma~\ref{lem_extendibleseparable} to obtain
   \begin{multline}
      \inf_{\sigma_{A C} \in S_{A : C}} \Delta(\rho_{A C}, \sigma_{A C}) 
  \leq \Delta(\rho_{A C}, \omega_{A C}) +   \inf_{\sigma_{A C} \in S_{A : C}} \Delta(\omega_{A C}, \sigma_{A C})   \\
  \leq  (k-1) \sqrt{\frac{\ln 2}{2} E_{\mathrm{sq}}(\rho_{A C})} + 2 \frac{(\dim C)^2}{k} \ .
  \end{multline}
  Inserting  $k = \ceil{ \sqrt[4]{\frac{8}{\ln(2) E_{\mathrm{sq}}(\rho_{A C})}}  \dim C }$ then yields the claim.     
   \comment{
   More precisely,  
   \begin{multline}
   (k-1) \sqrt{\frac{\ln(2)}{2} E_{\mathrm{sq}}(\rho_{A C}) } + 2 \frac{(\dim C)^2}{k}  \\
   \leq \sqrt[4]{\frac{8}{\ln(2) E_{\mathrm{sq}}(\rho_{A C})}}  \dim C \sqrt{\frac{\ln(2)}{2} E_{\mathrm{sq}}(\rho_{A C}) }  + 2 (\dim C)^2 \cdot \sqrt[4]{\frac{\ln(2) E_{\mathrm{sq}}(\rho_{A C})}{8}} \frac{1}{\dim C} \\
   =  2 \sqrt[4]{\frac{8}{\ln(2) E_{\mathrm{sq}}(\rho_{A C})}}  \dim C \sqrt{\frac{\ln(2)}{8} E_{\mathrm{sq}}(\rho_{A C}) }  + 2 \dim C \cdot \sqrt[4]{\frac{\ln(2) E_{\mathrm{sq}}(\rho_{A C})}{8}} \\
   = 4  \dim C \sqrt[4]{\frac{\ln(2) E_{\mathrm{sq}}(\rho_{A C})}{8}} \ .
   \end{multline}
  }
   \end{proof}   
 Corollary~\ref{cor_squashed} quantifies the faithfulness of squashed entanglement in terms of the trace norm. Compared to previously known versions of this claim~\cite{BCY11}, only the dimension of \emph{one} the two subsystems enters as a factor in the bound~\eqref{eq_squashedfaithful}. (Because of the symmetry of the other involved quantities, one can always choose the lower-dimensional one.) We also note that the example of the totally antisymmetric state on $A \otimes C$ with $\dim A = \dim C = d$ shows that such a factor is necessary. Indeed, the squashed entanglement of this state is of the order $O(1/d)$~\cite{CSW12} whereas its trace distance to the closest separable state cannot be smaller than~$\frac{1}{4}$. (To see this, note that for any product state $\sigma_A \otimes \sigma_C$ we have $\tr(\sigma_{A} \otimes \sigma_{C} \Pi_{\mathrm{as}}) \leq \frac{1}{2}$ where $\Pi_{\mathrm{as}}$ denotes the projector onto the antisymmetric subspace of $A \otimes C$.\comment{We have $\tr(\Pi_{\mathrm{as}} \sigma_{A} \otimes \sigma_{C}) = \frac{1}{2} \tr((\id_{AC} - F_{AC}) \sigma_{A} \otimes \sigma_{C}) = \frac{1}{2} \left(\tr(\sigma_A) \tr(\sigma_{C}) - \tr(\sigma_{A} \sigma_{C}) \right) \leq 1/2$.})
 
\section*{Acknowledgments}

We thank Aram Harrow for advice on Lemma~\ref{lem_SchurWeyldecomposition}, Andreas Winter for making us aware of the implications that our result has for squashed entanglement (described in Appendix~\ref{app_squashed}), and Mark Wilde for pointing out an error in the dimension factor of a previous version of Corollary~\ref{cor_squashed}. We also thank Normand Beaudry, Mario Berta, Fr\'ed\'eric Dupuis, Volkher Scholz, David Sutter, Marco Tomamichel, Mark Wilde, and Lin Zhang for discussions and comments on earlier versions of this manuscript. This project was supported by the European Research Council (ERC) via grant No.~258932 and from the Swiss National Science Foundation (SNSF) via the National Centre of Competence in Research ``QSIT''.

\bibliographystyle{plain}
\bibliography{big2}

\end{document}